\documentclass[11pt,a4paper]{article}
\usepackage{fullpage}
\usepackage{amsmath}
\usepackage{amssymb}
\usepackage{stmaryrd}
\usepackage{graphicx}
\usepackage{psfrag}
\usepackage{color}
\usepackage{subfigure}
\usepackage{relsize}
\usepackage{wasysym}
\usepackage[table]{xcolor}
\usepackage{soul}
\usepackage{hyperref}

 \newtheorem{definition}{Definition}
 
 \newtheorem{theorem}{Theorem}
 
 \newtheorem{lemma}{Lemma}
 \newenvironment{proof}{\noindent \begin{rm}{\textbf{Proof.} }}{\hspace*{\fill}$\Box$\par\end{rm}}

\newcommand{\fboxS}[1]{\fbox{\small{#1}}}

\newcommand{\id}{\mbox{\sf Id}}

\newcommand{\N}{\mbox{\sf N}}
\newcommand{\V}{\mbox{\sf V}}

\newcommand{\Bit}{\mbox{\rm Bit}}
\newcommand{\B}{\mbox{\sf dB}}

\newcommand{\leader}{\mbox{\sf leader}}
\newcommand{\dis}{\mbox{\sf d}}
\newcommand{\SB}{\mbox{$\widehat{\sf B}$}}
\newcommand{\dB}{\mbox{\sf dB}}

\newcommand{\ME}{\mbox{\sf Elec}}
\newcommand{\Phase}{\mbox{\sf Phase}}
\newcommand{\Bp}{\mbox{\sf Bit-Position}}
\newcommand{\Bc}{\mbox{\sf Control}}
\newcommand{\Bs}{\mbox{\sf Bit-Strong}}
\newcommand{\MV}{\mbox{\sf HC}}
\newcommand{\MC}{\mbox{\sf PL}}
\newcommand{\Add}{\mbox{\sf Add}}
\newcommand{\p}{\mbox{\sf p}}

\newcommand{\Child}{\mbox{\tt Ch}}
\newcommand{\Root}{\mbox{\tt Root}}
\newcommand{\Nd}{\mbox{\tt Nd}}
\newcommand{\inNd}{\mbox{\tt PassNd}}
\newcommand{\E}{\mbox{\tt Er}}
\newcommand{\Best}{\mbox{\tt Best}}

\newcommand{\NdReset}{\mbox{\tt maxElec}}

\newcommand{\maxSB}{\mbox{\tt Max\SB}}

\newcommand{\R}{\mathbb{R}}
\newcommand{\VChild}{\mbox{\tt VCh}}
\newcommand{\Up}{\mbox{\tt Up}}
\newcommand{\Pass}{\mbox{\tt Pass}}

\newcommand{\Error}{\tt T.Er}
\newcommand{\TAdd}{\tt T.Add}
\newcommand{\TUp}{\tt T.Update}
\newcommand{\TVerif}{\tt T.Verif}
\newcommand{\TBroad}{\tt T.Broad}
\newcommand{\TInc}{\tt T.Inc}
\newcommand{\TPass}{\tt T.Pass}
\newcommand{\TStart}{\tt T.Start}
\newcommand{\TStartdB}{\tt T.StartdB}
\newcommand{\TCleanM}{\tt T.CleanM}

\newcommand{\TReset}{\tt T.Reset}

\newcommand{\Neig}{\mbox{\tt Ng}}

\renewcommand{\Reset}{\mbox{$\mathcal Reset$}}
\newcommand{\Start}{\mbox{$\mathcal Start$}}
\newcommand{\Passive}{\mbox{$\mathcal Passive$}}
\newcommand{\Inc}{\mbox{$\mathcal Inc$}}
\newcommand{\StartdB}{\mbox{$\mathcal StartdB$}}
\newcommand{\BinaryAd}{\mbox{$\mathcal BinAdd$}}
\newcommand{\Broadcast}{\mbox{$\mathcal Broad$}}
\newcommand{\Verification}{\mbox{$\mathcal Verif$}}
\newcommand{\CleanM}{\mbox{$\mathcal CleanM$}}
\newcommand{\Update}{\mbox{$\mathcal Update$}}

\newcommand{\Algo}{\mbox{\bf{CLE}}}

\newcommand{\Rho}{\mbox{\text{P}}}

\newenvironment{smallitemize} {
 \begin{list}{$-$} {\setlength{\parsep}{0pt}
\setlength{\itemsep}{0pt}} } { \end{list} }

\newcommand{\remove}[1]{}

\definecolor{LLgray}{gray}{0.90}

\begin{document}
\title{Compact Deterministic Self-Stabilizing Leader Election:\\
The Exponential Advantage of Being Talkative\thanks{A preliminary version of this paper has appeared in  \cite{BlinT_PODC13,BlinT_DISC13}.}}

\author{
L{\'e}lia Blin\thanks{Additional support from the ANR project IRIS.} \\
\footnotesize{Universit\'{e} d'Evry-Val d'Essonne, 91000 Evry, France.}\\
\footnotesize{UPMC Sorbonne Universit\'{e}s, France.}\\
\footnotesize{LIP6-CNRS UMR 7606, France.}\\
\footnotesize{lelia.blin@lip6.fr}
\and S\'ebastien Tixeuil\\
\footnotesize{UPMC Sorbonne Universit\'{e}s, France.}\\
\footnotesize{Institut Universitaire de France.}\\
\footnotesize{LIP6-CNRS UMR 7606.}\\
\footnotesize{sebastien.tixeuil@lip6.fr}
}
\date{}
\maketitle
\begin{abstract}
This paper focuses on \emph{compact} deterministic self-stabilizing solutions for the leader election problem. 
When the protocol is required to be \emph{silent} (i.e., when communication content remains fixed from some point in time during any execution), there exists a lower bound of $\Omega(\log n)$ bits of memory per node participating to the leader election (where $n$ denotes the number of nodes in the system). This lower bound holds even in rings. We present a new deterministic (non-silent) self-stabilizing protocol for $n$-node rings that uses only $O(\log\log n)$ memory bits per node, and stabilizes in $O(n\log^2 n)$  rounds. Our protocol has several attractive features that make it suitable for practical purposes. First, the communication model fits with the model used by existing compilers for real networks. Second, the size of the ring (or any upper bound on this size) needs not to be known by any node. Third, the node identifiers can be of various sizes. Finally, no synchrony assumption, besides a weakly fair scheduler, is assumed. Therefore, our result shows that, perhaps surprisingly, trading silence for exponential improvement in term of memory space does not come at a high cost regarding  stabilization time or minimal assumptions.
\end{abstract}

\newpage

\section{Introduction}
\label{sec:intro}

This paper is targeting the issue of designing efficient self-stabilization algorithm for the leader election problem. \emph{Self-stabilization}~\cite{D74j,D00b,T09bc} is a general paradigm to provide forward recovery capabilities to distributed systems and networks. Intuitively, a protocol is self-stabilizing if it is able to recover  from any transient failure, without external intervention. \emph{Leader election} is one of the fundamental building blocks of distributed computing, as it enables to distinguish a single node in the system, and thus to perform specific actions using that node. Leader election is especially important in the context of self-stabilization as many protocols for various problems assume that a single leader exists in the system, even when faults occur. Hence, a self-stabilizing leader election mechanism enables to run such protocols in networks where no leader is a priori given, by using simple composition techniques~\cite{D00b}.

Most of the literature in self-stabilization is dedicated to improving efficiency after failures occur, including minimizing the stabilization time, i.e., the maximum amount of time one has to wait before  recovering from a failure. While stabilization time is meaningful to evaluate the efficiency of an algorithm in the presence of failures, it does not necessarily capture the overhead of self-stabilization when there are no faults~\cite{ANT12c}, or after stabilization. Another important criterion to evaluate this overhead is the \emph{memory space} used by each node. This criterion is motivated by two practical reasons, that we detail below. 

First, self-stabilizing protocols require that \emph{some} communications carry on forever (in order to be able to detect distributed inconsistencies due to transient failures~\cite{BDDT07j,DMT09c}). Therefore, minimizing the memory space used by each node enables to  minimize the amount of information that is exchanged between nodes. Indeed, protocols are typically written in the \emph{state model}, where the state of each node is read by each of its neighbors. (The use of the state model is motivated by the fact that all existing stabilization-preserving compilers~\cite{AK09c,CG12j,DMCDHSHLAG08c,MG05b} are precisely designed for this  model).  

Second, minimizing memory space enables to significantly reduce the cost of redundancy when mixing self-stabilization and replication, in order to increase the probability of masking or containing transient faults~\cite{GCH06c,HP00j}. For instance, duplicating every bit three times at each node permits to withstand one randomly flipped bit.  More generally, decreasing the memory space allows the designer to duplicate this memory many times, in order to tolerate many random bit-flips.

A foundational result regarding memory space in the context of self-sta\-bi\-li\-za\-tion is due to Dolev \emph{et al.}~\cite{DGS99j}, which  states that, $n$-node networks, $\Omega(\log n)$ bits of memory are required for solving global tasks such as leader election. Importantly, this bound holds even for the ring. A key component of this lower bound  is that the protocol is assumed to be \emph{silent}. (Recall that a protocol is silent if each of its executions reaches a point in time beyond which the registers containing the information available at each node do \emph{not} change). The lower bound can be extended to \emph{non-silent} protocols, but only for protocols with restricted capabilities. For instance, it holds in \emph{anonymous} (and uniform) unidirectional rings, even of prime size \cite{BGJ07j,FJ01c}. As a matter of fact, most deterministic self-stabilizing leader election protocols~\cite{AB98j,AG94j,AKMPV07j,DLV11j,DH97j} use at least $\Omega(\log n)$ bits of memory per node. Indeed, either these protocols directly compare node identifiers (and thus communicate node identifiers to neighbors), or they compute some variant of a hop-count distance to the elected node (and this distance can be as large as  $\Omega(n)$ to be accurate). 

A few previous work~\cite{IL94c,ILS95c,MOOY92c} managed to break the $\Omega(\log n)$ bits lower bound for the memory space of self-stabilizing leader election algorithms. Nevertheless, the corresponding algorithms exhibit shortcomings that hinder their relevance to practical applications. For instance, the algorithm by Mayer \emph{et al.}~\cite{MOOY92c}, by Itkis and Levin~\cite{IL94c}, and by Awerbuch and Ostrovstky~\cite{AO94c} use a constant number of bits per node only. However, these algorithms guarantee \emph{probabilistic} self-stabilization only (in the Las Vegas sense). In particular, the stabilization time is only \emph{expected} to be polynomial in the size of the network. Moreover, these algorithms are designed for a communication model that is more powerful than the classical state model used in this paper. (The state model is the model used in most available compilers for actual networks~\cite{AK09c,CG12j,DMCDHSHLAG08c,MG05b}). More specifically, Mayer \emph{et al.}~\cite{MOOY92c} use the message passing model, and Awerbuch and Ostrovsky~\cite{AO94c} use the link-register model, where communications between neighboring nodes are carried out through dedicated registers. Finally, Itkis and Levin~\cite{IL94c} use the state model augmented with reciprocal pointer to neighbors. In this model, not only a node $u$ is able to distinguish a particular neighbor $v$ (which can be done using local labeling), but also this distinguished neighbor $v$ is aware that it has been selected by $u$. Implementing this mutual interaction between neighbors  typically requires distance-two coloring, link coloring, or two-hops communication. All these techniques are impacting the memory space requirement significantly~\cite{MT09j}. It is also important to note that the communication models in~\cite{AO94c,IL94c,MOOY92c} allow nodes to send different information to different neighbors, while this capability is beyond the power of the classical state model. The ability to send different messages to different neighbors is a strong assumption in the context of self-stabilization. It anables to construct a ``path of information'' that is consistent between nodes. This path is typically used to distribute the storage of information along a path, in order to reduce the information stored at each node. However, this assumption prevents the user from taking advantage of the  existing compilers. So implementing the protocols in~\cite{AO94c,IL94c,MOOY92c}  to actual networks requires to rewrite all the codes from scratch. 

To our knowledge, the only \emph{deterministic} self-stabilizing leader election protocol using sub-logarithmic memory space in the classical state model is due to Itkis \emph{et al.}~\cite{ILS95c}. Their elegant algorithm uses only a constant number of bits per node,  and stabilizes in $O(n^2)$ time in $n$-node rings. However, the algorithm relies on several restricting assumptions. First, the algorithm works properly only if the size of the ring is \emph{prime}. Second, it assumes that, at any time, a \emph{single} node is scheduled for execution, that is, it assumes a \emph{central} scheduler~\cite{DT11r}. Such a scheduler is far less practical than the classical \emph{distributed} scheduler, which allows any set of processes to be scheduled concurrently for execution. Third, the algorithm in~\cite{ILS95c} assumes that the ring is \emph{oriented}. That is, every node is supposed to possess  a consistent notion of left and right. This orientation permits to mimic the behavior of reciprocal pointer to neighbors mentioned above. Extending the algorithm by Itkis \emph{et al.}~\cite{ILS95c} to more practical settings, i.e., to non-oriented rings of arbitrary size, to the use of a distributed scheduler, etc,  is not trivial if one wants to preserve a sub-logarithmic memory space at each node. For example, the existing transformers enabling to enhance protocols designed for the central scheduler in order to  operate under the distributed scheduler require $\Theta(\log n)$ memory at each node~\cite{DT11r}. Similarly, self-stabilizing ring-orientation protocols exist, but those preserving sub-logarithmic memory space either works only in rings of odd size for deterministic guarantees~\cite{H98j}, or just provide probabilistic guarantees~\cite{IJ93j}. Moreover, in both cases, the stabilization time is $O(n^2)$, which is quite large.

To summarize, all existing self-stabilizing leader election algorithm designed in a practical communication model, and for rings of arbitrary size, without a priori orientation, use $\Omega(\log n)$ bits of memory per node. Breaking this bound, without introducing any kind of restriction on the settings, requires, beside being non-silent, a completely new approach.  

\subsection*{Our results}

In this paper, we present a deterministic (non-silent) self-stabilizing leader election algorithm that operates under the distributed scheduler in non-anonymous undirected  rings of arbitrary size. Our  algorithm is non-silent to circumvent the lower bound $\Omega(\log n)$ bits of memory per node in~\cite{DGS99j}. It uses only $O(\log \log n)$ bits of memory per node, and stabilizes in $O(n \log^2n)$ time. 

Unlike the algorithms in~\cite{AO94c,IL94c,MOOY92c}, our algorithm is deterministic, and designed to run under the classical state-sharing communication model (a.k.a. state model), which allows it to be implemented by using actual compilers~\cite{AK09c,CG12j,DMCDHSHLAG08c,MG05b}. Unlike~\cite{ILS95c}, the size of the ring is arbitrary, the ring is not assumed to be oriented, and the scheduler is distributed. Moreover the stabilization time of our algorithm is significantly smaller than the one in \cite{ILS95c}. Similarly to~\cite{AO94c,IL94c,KormanKM11,MOOY92c}, our algorithm uses a technique to distribute the information among nearby nodes along a sub-path of the ring. However, our algorithm does not rely on powerful communication models such as the ones used in~\cite{AO94c,IL94c,MOOY92c}. Those powerful communication models make easy the construction and management of such sub-paths. The use of the classical state-sharing model makes the construction and management of the sub-paths much more difficult. 

Besides the use of a sub-logarithmic memory space, and beside a quasi-linear stabilization time, our algorithm possesses several attractive features.  First, the size (or any  upper bound for this size) need not to be known by any node. Second, the node identifiers (or identities) can be of various sizes (to model, e.g., Internet networks running different versions of IP). Third, no synchrony assumption besides weak fairness is assumed (a node that is continuously enabled for execution is eventually scheduled for execution). 

At a high level, our algorithm is essentially based on two techniques. One consists in electing the leader by comparing the identities of the nodes, bitwise, which requires special care, especially when the node identities can be of various sizes. The second technique consists in maintaining and merging trees based on a parenthood relation, and verifying the absence of cycles in the 1-factor induced by this parenthood relation. This verification is performed using small memory space by grouping the nodes in hyper-nodes of appropriate size. Each hyper-node handles an integer encoding a distance to a root. The bits of this distance are  distributed among the nodes of the hyper-nodes to preserve a small memory per node. Difficulties arise when one needs to perform arithmetic operations on these distributed bits, especially in the context where nodes are unaware of the size of the ring. The precise design of our algorithm requires overcoming many other difficulties due to the need of maintaining correct information in an environment subject to arbitrary faults. 

In addition, we took care of designing our algorithm to be ready for implementation, i.e., we do not only describe a conceptual protocol, but also produce a \emph{concrete} self-stabilizing leader election protocol. This article provides a high level description of our algorithm,   a detailed description of the protocol, and a complete proof of correctness.\footnote{Moreover, the reader is invited to consult \href{http://www-npa.lip6.fr/~blin/Election/}{\textcolor{blue}{\sl www-npa.lip6.fr/\~{}blin/Election/}} where a video of the dynamic execution of the protocol is presented. This video is the result of a complete implementation of the protocol. The video execution is using a randomized distributed scheduler. The initial configuration is illegitimate, and the video displays the protocol operation towards a legitimate configuration. } 
To sum up, our result shows that, perhaps surprisingly, trading silence for exponential improvement in term of memory space does not come at a high cost regarding  stabilization time, neither it does regarding minimal assumptions about the communication framework. 


\section{Model and definitions}
\label{sec:model}

\subsection{Program syntax and semantics} 

A distributed system consists of $n$ processors that form a
communication graph. The processors are represented by the nodes of this graph, and the edges represent pairs of processors
that can communicate directly with each other. Such processors are said to be \emph{neighbors}. This classical model is called \emph{state-sharing communication model}.
The \emph{distance} between two processors is the length (i.e., number of edges) of the shortest
path between them in the communication graph. Each processor contains
variables, and rules. A variable ranges over a fixed domain of values.
A rule is of the form $$\langle label \rangle : \langle guard \rangle
\longrightarrow \langle command \rangle.$$  A \emph{guard} is a boolean
predicate over processor variables. A \emph{command} is a set of
variable-assignments. A command of processor $p$ can only update its own variables.  
On the other hand, $p$ can read the variables of its neighbors. 
An assignment of values to all variables in the system is called a
\emph{configuration}. A rule whose guard is \textbf{true} in some
system configuration is said to be \emph{enabled} in this configuration. The rule
is \emph{disabled} otherwise. An atomic execution of a subset of
enabled rules results in a transition of the system from one configuration to
another. This transition is called a \emph{step}. A \emph{run}
of a distributed system is a sequence of  transitions. 

\subsection{Schedulers}

A \emph{scheduler}, also called \emph{daemon}, is a restriction on the
runs to be considered. The schedulers differ among them by different execution semantics,
and by different fairness in the activation of the processors~\cite{DT11r}. 
With respect to execution semantics, we consider  the least restrictive
scheduler, called the \emph{distributed scheduler}.
In the run of a distributed scheduler, a step can contain the execution of an
arbitrary subset of enabled rules of correct processors. 
With respect to fairness, we use the least restrictive scheduler, called \emph{weakly fair scheduler}. In every run of the weakly fair
scheduler, a rule of a correct processor is executed infinitely often if it is enabled in all
but finitely many configurations of the run. That is, the rule has to
be executed only if it is continuously enabled.  A \emph{round} is the smallest portion of an execution where every process has the opportunity to execute at least one action.

\subsection{Predicates and specifications}
 
A predicate is a boolean function over network configurations. A
configuration \emph{conforms} to some predicate $R$, if $R$ evaluates
to \textbf{true} in this configuration. The configuration
\emph{violates} the predicate otherwise.  Predicate $R$ is
\emph{closed} in a certain protocol $P$, if every configuration of
a run of $P$ conforms to $R$, provided that the protocol starts
from a configuration conforming to $R$. Note that if a protocol
configuration conforms to $R$, and the  configuration resulting from 
the execution of any step of $P$ also conforms to $R$, then $R$
is closed in $P$.

A \emph{specification} for a processor $p$ defines a set of
configuration sequences. These sequences are formed by variables of
some subset of processors in the system. This subset always includes
$p$ itself. A \emph{problem specification}, or  \emph{problem} for short,
defines specifications for each processor of the system. A problem
specification in the presence of faults defines specifications for
correct processors only.  Program $P$ \emph{solves} problem
$S$ under a certain scheduler if every run of $P$ satisfies
the specifications defined by $S$.  A closed predicate $I$ is an
\emph{invariant} of program $P$ with respect to problem $S$
if every run of $P$ that starts in a state conforming to $I$
satisfies $S$. 
Given two predicates $l_1$ and $l_2$ for program $P$ with respect to
problem $S$, $l_2$ is an \emph{attractor} for $l_1$ if every run that
starts from a configuration that conforms to $l_1$ contains a
configuration that conforms to $l_2$. Such a relationship is denoted
 by $$l_1 \triangleright l_2.$$
A program $P$ is \emph{self-stabilizing} \cite{D74j} to specification
$S$ if every run of $P$ that starts in an arbitrary
configuration contains a configuration conforming to an invariant of
$P$ with respect to problem $S$. That is, this invariant is an attractor of predicate \emph{true}.

\subsection{Leader election specification}
 
Consider a system of processors where each processor has a boolean
variable $\leader$. We use the classical definition of 
\emph{leader election}, which specifies that, in every protocol run, 
there is a suffix where a single processor $p$ has $\leader_p=true$, and 
every other processor $q\neq p$ satisfies $\leader_q=false$.


\section{A compact leader-election protocol for rings}
\label{sec:protocol}

In this section, we describe our self-stabilizing algorithm for leader election in arbitrary $n$-node rings. The algorithm is later proved to use $O(\log\log n)$ bits of memory per node, and to stabilize in quasi-linear time, whenever the identities of the nodes are between 1 and $n^c$, for some $c\geq 1$. For the sake of simplicity, we assume that the identifiers are in $[1,n]$. Nevertheless, the algorithm works without assuming any particular range for the identifiers. We first provide a general overview of the algorithm, followed by a more detailed description in Section~\ref{subsec:detaildescription}. All predicates and commands are postponed in section~\ref{subsec:Pred}.

\subsection{Overview of the algorithm}

As  many existing deterministic self-stabilizing leader election algorithms, our algorithm aims at electing the node with maximum identity among all nodes, and, simultaneously, at constructing a spanning tree rooted at the elected node. The main constraint imposed by our wish to use sub-logarithmic memory is that we cannot exchange or even locally use complete identifiers, as their size $\Omega(\log n)$ bits does not fit in a sub-logarithmic size memory. As a matter of fact, we assume that every node can access the bits of its identifier, but only a constant number of them can be simultaneously stored and/or communicated to neighbors at any given time. Our algorithm makes sure that every node stores the current position of a particular bit of the identifier, referred to as a \emph{bit-position}  in the sequel. 

\subsubsection{Selection of the leader}

Our algorithm operates in phases. At each phase, each node that is a candidate leader $v$ reveals some bit-position, different from the ones at the previous phases, to its neighbors. More precisely, let $\id_v$ be the identity of node $v$, and assume that $\id_v=\sum_{i=0}^k b_i 2^i$. Let $I=\big \{i\in\{0,...,k\},b_i\neq 0\big\}$ be the set of all non-zero bit-positions in the binary representation of $\id_v$. Let us rewrite $I=\{p_1,...,p_j\}$ with $0\leq p_1<p_2<...<p_j\leq k$. Then, during Phase~$i$, $i=1,\dots,j$, node $v$ reveals $p_{j-i+1}$ to its neighbors, which potentially propagate it to their neighbors, and possibly to the whole network in subsequent phases. During Phase $i$, for $j+1\leq i\leq \lfloor \log n\rfloor +1$, node $v$ either becomes passive (that is, stops acting as a candidate leader) or remains a candidate leader. If, at the beginning of the execution of the algorithm, all nodes are \emph{candidate} leaders, then during each phase, some candidate leaders are eliminated, until exactly one candidate leader remains, which becomes the actual leader. More precisely, let $p_{max}(i)$ be the most significant bit-position revealed at Phase~$i$ among all nodes. Then, among all candidate leaders still competing for becoming leader, only those whose bit-position revealed at Phase~$i$ is equal to $p_{max}(i)$  carry on the electing process. The other ones become passive. 

If all identities are in $[1,n]$, then the communicated bit-positions are less than $\lceil\log n\rceil$, and thus can be represented with $O(\log\log n)$ bits. The difficulty is to implement this simple ``compact'' leader election mechanism in a self-stabilizing manner. In particular, the nodes may not have same number of bits  encoding their identifiers, the ring may not start from a configuration where every node is a candidate leader, and the distributed scheduler may lead nodes to operate at various paces.

An additional problem in self-stabilizing  leader election is the potential presence of \emph{impostor} leaders. If one can store the identity of the leader at each node, then detecting an impostor is easy. Under our memory constraints, nodes cannot store the  identity of the leader, nor read entirely their own identifier. So, detecting impostor leaders becomes non trivial, notably when an impostor has an identity whose most significant bit is equal to the most significant bit of the leader. To overcome this problem, the selection of the leader must run perpetually, leading our algorithm to be non-silent. 

\subsubsection{Spanning tree construction} 

Our approach to make the above scheme self-sta\-bi\-li\-zing is to merge the leader election process with a tree construction process. Every candidate leader is the root of a tree. Whenever a candidate leader becomes passive, its tree is merged to another tree, until there remains only one tree.
The main obstacle in self-stabilizing tree-construction is to handle an arbitrary initial configuration. This is particularly difficult if the initial configuration yields a cycle rather than a spanning forest. In this case, when the leader election subroutine, and the tree construction subroutine are conjointly used, the presence of the cycle implies that, while every node is expecting to point to a neighbor leading to a leader, there are no leaders in the network. Such a configuration is called  \emph{fake} leader. In order to break cycles that can be present in the initial configuration, we use an improved variant of the classical distance calculation~\cite{DIM97j}. In the classical approach, every node $u$ maintains an integer variable $\dis_u$ that  stores the distance from $u$ to the root of its tree. If $v$ denotes the parent of $u$, then typically $\dis_v=\dis_u-1$, and if $\dis_v \geq \dis_u$, then $u$ deletes its pointer to $v$. If the topology of the network is a ring, then detecting the presence of an initial spanning cycle, instead of a spanning forest, may involve distance variables as large as $n$, inducing $\Omega(\log n)$ bits of memory. 

In order to use exponentially less memory, our algorithm uses the distance technique but modulo $\log n$. More specifically, each node $v$ maintains three variables. The first variable is an integer denoted by $\dis_v \in\{0,...,\lfloor \log n \rfloor\}$, called the ``distance'' of node $v$. Only candidate leaders $v$ can have $\dis_v=0$. Each node $v$ maintains $\dis_v=1+(\min\{\dis_{u},\dis_{u'}\} \bmod \lfloor \log n \rfloor)$
where $u$ and $u'$ are the neighbors of $v$ in the ring.  Note that  nodes are not aware of $n$. Thus they do not actually use the value $\lfloor \log n\rfloor$ as above, but a potentially erroneous estimation of it. 

The second variable is $\p_v$, denoting the parent of node $v$. This parent is its neighbor $w$ such that  $\dis_v=1 + (\dis_w \bmod  \lfloor \log n \rfloor)$. By itself, this technique is not sufficient to detect the presence of a cycle, because the number of nodes can be a multiple of $\lfloor \log n\rfloor$. Therefore, we also  introduce the notion of \emph{hyper-node}, defined as follows: 

\begin{definition}
A \em{hyper-node} $X$ is a set  $\{x_1,x_2,\cdots,x_{\lfloor \log n \rfloor}\}$ of consecutive  nodes in the ring, such that \mbox{$\dis_{x_1}=1$}, \mbox{$\dis_{x_2}=2$},..., \mbox{$\dis_{x_{\lfloor \log n \rfloor}}=\lfloor \log n \rfloor$}, $\p_{x_2} =x_1,$ $\p_{x_3}=x_2,...,$ $\p_{x_{\lfloor \log n \rfloor}}=x_{\lfloor \log n \rfloor-1}$ and $\p_{x_1} \neq x_2$. 
\end{definition}

The parent of a hyper-node $X=\{x_1,x_2,\cdots,x_{\lfloor \log n \rfloor}\}$ is a hyper-node $Y=\{y_1,y_2,\cdots,y_{\lfloor \log n \rfloor}\}$ such that $\p_{x_1} =y_{\lfloor \log n \rfloor}$. By definition, there are at most $\lceil n/\lfloor \log n \rfloor \rceil $ hyper-nodes. If $n$ is not divisible by $\lfloor \log n \rfloor$, then some nodes can be elements of an incomplete hyper-node. There can be several incomplete hyper-nodes, but if the parent of a (complete) hyper-node  is an incomplete hyper-node, then an error is detected. Incomplete hyper-nodes must be leaves: there cannot be incomplete hyper-nodes in a cycle. 

The key to our protocol is that hyper-nodes can maintain larger distance information than simple nodes, by distributing the information among the nodes of a hyper-node. More precisely, we assume that each node $v$ maintains a bit of information, stored in variable $\dB_v$. Let $X=\{x_1,x_2,\cdots,x_{\lfloor \log n \rfloor}\}$ be a hyper-node, the set $\B_X=\{\dB_{x_1},\dB_{x_2},\cdots,$ $\dB_{x_{\lfloor \log n \rfloor}}\}$ can be considered as the binary representation of an integer on ${\lfloor \log n \rfloor}$ bits, i.e., between $0$ and $2^{\lfloor \log n \rfloor}-1$. Now, it is possible to use the same distance approach as usual, but at the hyper-node level. Part of our protocol consists in comparing, for two hyper-nodes $X$ and $Y$, the distance $\B_X$ and the distance $\B_Y$. If $Y$ is the parent of $X$, then the difference between $\B_X$ and $\B_Y$ must be one. Otherwise an inconsistency is detected regarding the current spanning forest. The fact that hyper-nodes include ${\lfloor \log n \rfloor}$ nodes implies that dealing with distances between hyper-nodes  is sufficient to detect the presence of a cycle spanning the $n$-node ring. This is because $2^{\lfloor \log n \rfloor} \geq n/\log n$.  (Note that hyper-nodes with $k$ nodes such that $2^k\geq n/k$ would do the same). 

In essence, the part of our algorithm dedicated to checking the absence of a spanning cycle generated by the parenthood relation boils down to comparing distances between hyper-nodes. Note that comparing distances between hyper-node involves communication at distance $\Omega(\log n)$. This is another reason why our algorithm is non-silent. 

\subsection{Detailed description}
\label{subsec:detaildescription}

\subsubsection{Notations and preliminaries}

Let $C_n=(V,E)$ be the $n$-node  ring, where $V$ is the set of nodes, and $E$ the set of edges. 
A node $v$ has access to an unique identifier, but can only access to this identifier one bit at a time, using the $\Bit(x,v)$ function, that returns the position of the $x$th most significant bit equal to $1$ in $\id_v$. This position can be encoded with $O(\log \log n)$ bits when identifiers are encoded using $O(\log n)$ bits, as we assume they are. 
A node $v$ has access to local port number associated to its adjacent edges. The variable parent of node $v$, denoted by $\p_v$, is actually the port number of the edge connecting $v$ to its parent. In case of $n$-node rings, $\p_v\in\{0,1\}$ for every $v$. (We do not assume any consistency between the port numbers). In a legitimate configuration, the structure induced by the parenthood relation must be a tree. The presence of more than one tree, or of a cycle, correspond to illegitimate configurations. We denote by $\N_v$ the set of the neighbors of $v$ in $C_n$, for any node $v \in V$. 

The variable distance, denoted by $\dis_v$ at node $v$, takes values in $\{-1,0,1,...,$ $\lfloor \log n\rfloor\}$. We have  $\dis_v=-1$ if all the variables of $v$ are reset. We have  $\dis_v=0$ if the node $v$ is a root of some tree induced by the parenthood relation. Such a root is also called candidate leader. Finally, $\dis_v\in\{1,...,\lfloor \log n\rfloor\}$ if $v$ is a node of some tree induced by the parenthood relation, different from the root. Such a node is also called passive.  Note that we only assume that variable $\dis$ can hold \emph{at least} (and not \emph{exactly}) $\lfloor \log n\rfloor +1$ different values, since nodes are not aware of how many they are in the ring, and just use an estimation of $n$.
The children of a node $v$ are returned by the macro $\Child(v)$, which returns the port number(s) of the edges leading to the child(ren) of $v$.

To detect cycles, we use four variables. First, each node maintains the variable $\dB$ introduced in the previous section, for constructing a distributed integer stored on an hyper-node. The second variable, $\Add_v\in\{+,\text{ok},\emptyset\}$, is used for performing additions involving values stored distributively on hyper-nodes. The third variable, $\MC_v$ (for pipeline), is used to send the result of an addition to the hyper-node children of the hyper-node containing $v$. Finally, the fourth  variable, $\MV_v$ (for Hyper-node Checking), is dedicated to checking the hyper-node bits. Variables $\MC_v$ and $\MV_v$ are either empty, or each composed of a pair of variables $(x,y)\in \{1,...,\lfloor \log n \rfloor\}\times \{0,1\}$. 

For constructing the tree rooted at the node with highest identity, we use three additional variables. After convergence, we expect the leader to be the unique node with distance zero, and to be the root of an inward directed spanning tree of the ring, where the arc of the tree is defined by the parenthood relation.  To satisfy the leader election specifications, we introduce the variable $\leader_v\in\{0,1\}$ whose value is $1$ if $v$ is the leader and $0$ otherwise. Since we do not assume that the identifiers of every node are encoded on the same number of bits, simply comparing the $i$-th most significant bit of two nodes is irrelevant. So, we use variable $\SB$, which represents the most significant bit-position of all the identities present in the ring. This variable is also locally used at each node $v$ as an estimate of $\lfloor \log n \rfloor$. Only the nodes $v$ whose variable $\SB_v$ is equal to the most significant bit of the $\id_v$ carry on participating to the election. Finally, the variables $\Bs$, $\Phase$, $\Bp$ and $\Bc$ are the core of the election process. Let $r$ be the root of the tree including node $v$. Then, the variable $\Bs_v$  stores the position of the most significant bit of $\id_r$, variable $\Phase_v$ stores the current phase number $i$, variable $\Bp$ stores the bit-position of $\id_r$ at phase $i$, and variable $\Bc_v$ 
stores a boolean  dedicated to the control of the updating of the elections variables.

\subsubsection{The Compact Leader Election algorithm $\Algo$}

\begin{figure}[tb]
\fboxS{
\footnotesize
$
\begin{array}{lcllll}
\R_{\tt Error}\hspace*{-2ex}&:&\hspace*{-2ex}  \hspace*{1,5ex}\Error(v) \vee \hspace*{1,5ex}\TReset(v) & \rightarrow& \Reset(v);\\\R_{\tt Start}\hspace*{-2ex}& :&\hspace*{-2ex}\neg  \Error(v) \wedge \neg\TReset(v)\wedge (\dis_v=\text{-}1)  \wedge \TStart(v)& \rightarrow & \Start(v);\\
\R_{\tt Passive}\hspace*{-2ex}&: &\hspace*{-2ex}\neg  \Error(v) \wedge \neg\TReset(v)\wedge (\dis_v>\text{-}1) \wedge \TPass(v)& \rightarrow & \Passive(v);\\
\R_{\tt Root}\hspace*{-2ex}&: &\hspace*{-2ex}\neg  \Error(v) \wedge \neg\TReset(v)\wedge (\dis_v=0)\hspace*{1ex} \wedge  \TPass(v)  \wedge \TStartdB(v) & \rightarrow & \StartdB(v);\\
\hspace*{-2ex}& &\hspace*{-2ex}\neg  \Error(v) \wedge \neg\TReset(v)\wedge (\dis_v=0)\hspace*{1ex}\wedge  \neg \TPass(v)  \wedge \TInc(v) & \rightarrow & \Inc(v);\\
\R_{\tt Update}\hspace*{-2ex}&: &\hspace*{-2ex}\neg  \Error(v) \wedge \neg\TReset(v)\wedge (\dis_v>0)\hspace*{1ex}\wedge  \neg \TPass(v)   \wedge  \neg{\tt EqElecP}(v) \wedge  \TUp(v) & \rightarrow & \Update(v);\\
\R_{\tt HyperNd}\hspace*{-2ex}&: &\hspace*{-2ex}\neg  \Error(v) \wedge \neg\TReset(v)\wedge (\dis_v>0)\hspace*{1ex}\wedge  \neg \TPass(v)  \wedge {\tt EqElecP}(v) \wedge& &\\
\hspace*{-2ex}& &\hspace*{-2ex}\hspace*{5ex}  (\Add_v=\emptyset) \wedge \TAdd(v) & \rightarrow & \BinaryAd(v);\\
\hspace*{-2ex}& &\hspace*{-2ex}\hspace*{5ex}  (\Add_v\neq\emptyset) \wedge  \TBroad(v) & \rightarrow & \Broadcast(v);\\
\hspace*{-2ex}& &\hspace*{-2ex}\hspace*{5ex}  \TVerif(v) & \rightarrow & \Verification(v);\\
\hspace*{-2ex}& &\hspace*{-2ex}\hspace*{5ex}  \TCleanM(v) & \rightarrow & \CleanM(v);\\
\end{array}
$
}

\caption{\sl Formal description of algorithm $\Algo$.  }
\label{fig:Algo}
\end{figure}

Algorithm $\Algo$ is presented in Figure~\ref{fig:Algo}. In this figure, a rule of the form 
\[
label : guard_0 \wedge (guard_1 \vee guard_2) \longrightarrow (command_1 ; command_2)
\]
where $command_i$ is performed when $guard_0\wedge guard_i$ is true. Such a rule is presented in several lines, one for the common guard, $guard_0$, and one for each alternative guards, $guard_i$, with their respective command. Figure~\ref{fig:Algo} describes the rules of the algorithm. 

\medskip

\noindent  $\Algo$ is composed of six rules:

\begin{smallitemize}

\item The rule $\R_{\tt Error}$, detects at node $v$ the presence of inconsistencies between the content of its variables and the content of its neighboring variables. If $v$ has not reset its variables, or has not restarted, the command $\Reset(v)$ is activated, i.e., all the content of all the variables at node $v$ are reset, and the variable  $\dis_v$ is set to $-1$. 

\item The rule $\R_{\tt Start}$, makes sure that, if an inconsistency is detected at some node, then all the nodes of the network reset their variables, and restart.  Before restarting, every node $v$ waits until all its neighbors are reset or have restarted. A node $v$ that restarts sets $\dis_v=0$, and its election variables $\Bs_v$, $\Bp_v$ appropriately, with $\Phase_v=1$. 

\item The rule $\R_{\tt Passive}$, is dedicated to the election process. A node $v$ uses command $\Passive(v)$ when one of its neighbors has a bit-position larger than its bit-position, at the same phase. 

\item The  rule $\R_{\tt Root}$, concerns the candidate leaders,  i.e.,  every node $v$ with $\dis_v=0$. Such a node can only execute the rule $\R_{\tt Root}$, resulting in that node performing one of the following two commands. Command $\StartdB(v)$ results in $v$ distributing the bit $\dB$ to its neighboring hyper-nodes. Command $\Inc(v)$ results in node $v$ increasing its phase by~1. This happens  when all the nodes between $v$ and others candidate leaders in the current tree are in the same phase, with the same election values $\Bs_v,\Bp_v,\Bc_v$. 

\item The  rule $\R_{\tt Update}$, is dedicated to updating the election variables. 

\item The rule $\R_{\tt HyperNd}$, is dedicated to the hyper-nodes distance verification. 

\end{smallitemize}

\subsubsection{Hyper-nodes distance verification}

Let us consider two hyper-nodes $X$ and $Y$ with $X$ the parent of $Y$. Our technique for verifying the distance between the two hyper-nodes $X$ and $Y$, is the following (see an example on Figure~\ref{fig:HypVerf}). $X$ initiates the verification. For this purpose, $X$ dedicates two local variables at each of its nodes: $\Add$ (to perform the addition) and $\MC$ (to broadcast the result of this addition inside $X$). Similarly, $Y$ uses the variable $\MV$ for receiving the result of the addition.

The binary addition starts at the node holding the last bit of $X$, that is node $x_k$ with $k=\SB_v$. Node $x_k$ sets $\Add_{x_k} := +$. Then, every node in $X$, but  $x_k$, proceeds as follows. For $k'<k$, if the child $x_{k''}$ of $x_{k'}$ has $\Add_{x_{k''}}=+$ and $\dB_{x_{k''}}=1$, then $x_k'$ assigns $+$ to $\Add_{x_{k''}}$. Otherwise,  if $\Add_{x_{k''}}=+$ and $\dB_{x_{k''}}=0$, the binary addition at this point does not generate a carry, and thus $x_{k'}$ assigns ``ok'' to $\Add_{x_{k'}}$. Since $\Add_{x_{k'}}=ok$, the binary addition is considered finished, and $x_{k'}$'s ancestors (parent, grand-parent, etc.) in the hyper-node assign ``ok'' to their variable $\Add$. However, if the first bit of $X$ (that is, $\dB_{x_1}$) is equal to one, then the algorithm detects an error because the addition would yield to an overflow. The result of the hyper-node binary addition is the following: if a node $x_k$ has $\Add_{x_k}=ok$, then it means that node $y_k$ holds the appropriate bit corresponding to the correct result of the addition  if and only if $\dB_{y_k}=\dB_{x_k}$. Otherwise, if $\Add_{x_k}=+$, then the bit at $y_k$ is correct if and only if $\dB_{y_k}=\overline{\dB}_{x_k}$\footnote{If $\dB_x=0$ then $\overline{\dB}_{x}=1$, and  if $\dB_x=1$ then $\overline{\dB}_{x}=0$.}.

\begin{figure}[tb]
\begin{center}
\begin{tabular}{ll}
\hspace*{-0,3cm}{\small(a)} &\hspace*{-0,45cm}\raisebox{-1,7ex}{\includegraphics[scale=0.315]{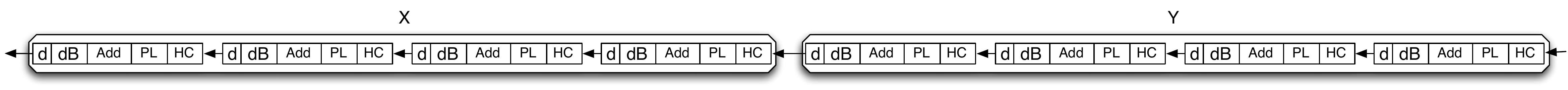}}\\
\hspace*{-0,3cm}{\small(b)} &\hspace*{-0,45cm}\raisebox{-1,7ex}{\includegraphics[scale=0.315]{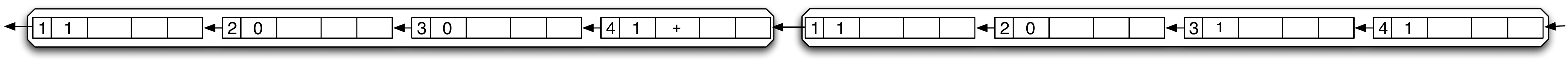}}\\
\hspace*{-0,3cm}{\small(c)} &\hspace*{-0,45cm}\raisebox{-1,7ex}{\includegraphics[scale=0.315]{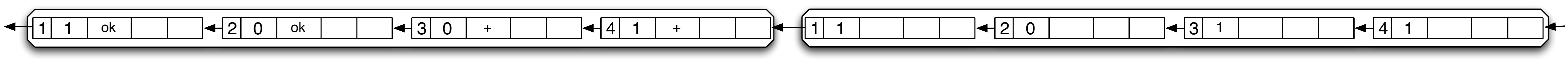}}\\
\hspace*{-0,3cm}{\small(d)} &\hspace*{-0,45cm}\raisebox{-1,7ex}{\includegraphics[scale=0.315]{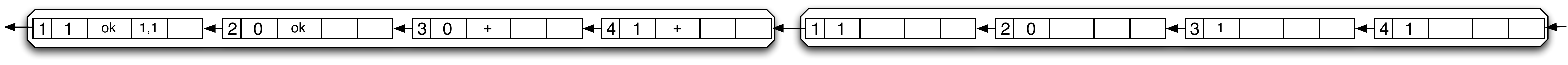}}\\
\hspace*{-0,3cm}{\small(e)} &\hspace*{-0,45cm}\raisebox{-1,7ex}{\includegraphics[scale=0.315]{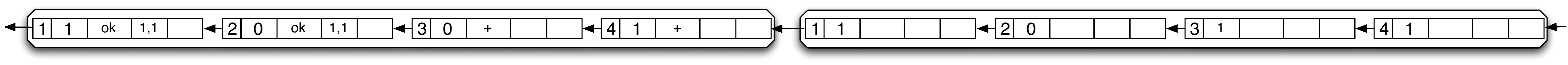}}\\
\hspace*{-0,3cm}{\small(f)} &\hspace*{-0,45cm}\raisebox{-1,7ex}{\includegraphics[scale=0.315]{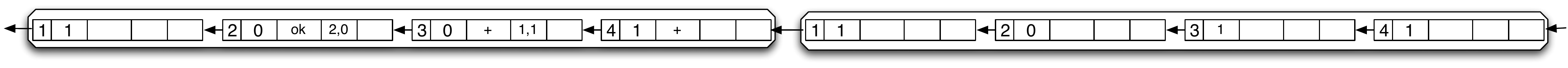}}\\
\hspace*{-0,3cm}{\small(g)} &\hspace*{-0,45cm}\raisebox{-1,7ex}{\includegraphics[scale=0.315]{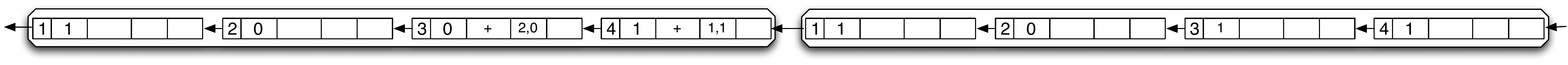}}\\
\hspace*{-0,3cm}{\small(h)} &\hspace*{-0,45cm}\raisebox{-1,7ex}{\includegraphics[scale=0.315]{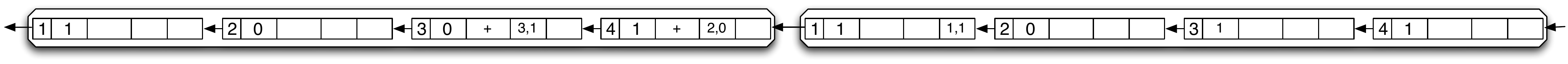}}\\
\hspace*{-0,3cm}{\small(i)} &\hspace*{-0,45cm}\raisebox{-1,7ex}{\includegraphics[scale=0.315]{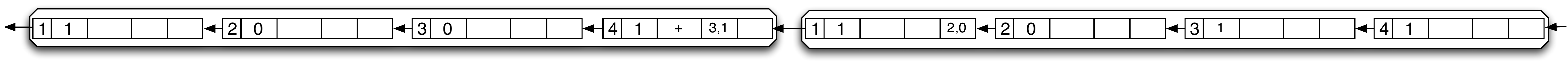}}\\
\hspace*{-0,3cm}{\small(j)} &\hspace*{-0,45cm}\raisebox{-1,7ex}{\includegraphics[scale=0.315]{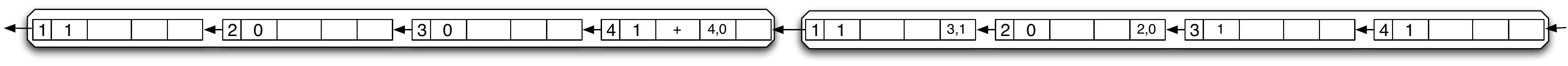}}\\
\hspace*{-0,3cm}{\small(k)} &\hspace*{-0,45cm}\raisebox{-1,7ex}{\includegraphics[scale=0.315]{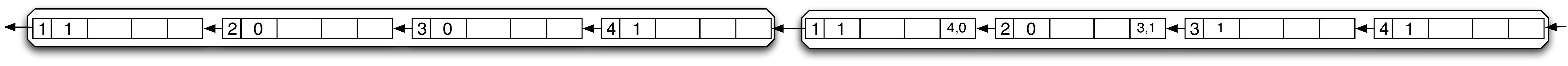}}\\
\hspace*{-0,3cm}{\small(l)} &\hspace*{-0,45cm}\raisebox{-1,7ex}{\includegraphics[scale=0.315]{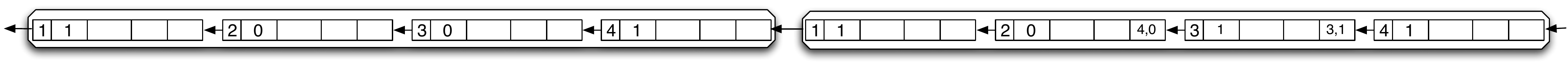}}\\
\hspace*{-0,3cm}{\small(m)} &\hspace*{-0,45cm}\raisebox{-1,7ex}{\includegraphics[scale=0.315]{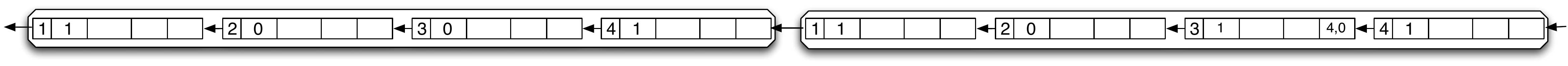}}\\
\hspace*{-0,3cm}{\small(n)} &\hspace*{-0,45cm}\raisebox{-1,7ex}{\includegraphics[scale=0.315]{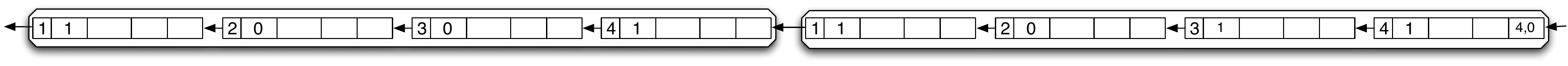}}\\
\end{tabular}
\end{center}
\caption{\footnotesize An example of distance verification between the hyper-node $X$ and its child $Y$. In this example, hyper-nodes are composed of four nodes. (a) The memory of each node is represented by three boxes storing, respectively, the distance of the node to the root, the bit-distance of the hyper-node the node belongs to, the binary addition information, plus two boxes containing information for the bit verification. (Empty boxes contains $\bot$).  (b) The last node of $X$ starts the addition. (c) The addition in $X$ is completed. (d) The first node of $X$ starts the verification. (e-g) The result of the addition is pipelined. (h) The first node $v$ of $Y$ checks $\dB_v$. (j) The second node $v'$ of $Y$ checks $\dB_{v'}$. (l) The third node $v''$ of $Y$ checks $\dB_{v''}$. (n) The last node of $Y$ detects an error.} 
\label{fig:HypVerf}
\end{figure}

The binary addition in $X$ is completed when node $x_1$ satisfies $\Add_{x_1}=+$ or \mbox{$\Add_{x_1}=ok$}. In that case, $x_1$ starts broadcasting the result of the addition. For this purpose, it sets \mbox{$\MC_{x_1}=(1,\dB_{y_1})$} where $\dB_{y_1}$ is obtained from $\Add_{x_1}$ and $\dB_{x_1}$.
Each node $x_i$ in $X$, $i > 1$, then successively perform the same operation as $x_1$. While doing so, node $x_i$ sets $\Add_{x_i}=\bot$, in order to enable the next verification. When the child of a node $x_i$ publishes $(\dis_{x_i},\dB_{x_i})$, node $x_i$ deletes $\MC_{x_i}$, in order to, again, enable the next verification. From $i=1,\dots,k$, all variables $\MC_{x_i}$  in $X$ are deleted.   
When $y_i$ sets $\MV_{y_i}[0]=\dis{y_i}$, node $y_i$ can check whether the bit in $\MV_{y_i}[1]$ corresponds to $\dB_{y_i}$. If yes, then the verification carries on. Otherwise $y_i$ detects a fault.

\subsubsection{Leader election and tree construction}

As previously mentioned, our leader election protocol simultaneously performs, together with the election of a leader,  the construction a tree  rooted at the leader. The leader should be the node whose identifier is maximal among all nodes in the ring. Our assumptions regarding identifiers are very weak. In particular, identifiers may be of various sizes, and the total number $n$  of different identifiers is not known to the nodes. In our algorithm, we use the variable $\SB$ to estimate (to some extent) the logarithm of the network size, and the variable $\Bs_v$ to propagate this estimation in the ring. More precisely, $\SB$ represents the most significant bit-position among all identities present in the ring, and we consider that all the nodes that do not carry the right value of $\SB$ in their local variables are not consistent. During the execution of the algorithm, only nodes whose identifiers match the most significant bit-position remain candidate leaders. Moreover, only candidate leaders broadcast bit-position during subsequent phases. 

\bigskip  

Let us  now detail the usage of the variables used for the election. For the sake of simplification, we introduce the variable $\ME$. 
The variable  $\ME$ is equal to a 4-tuple: 
\[
\ME_v=(\Bs_v,\Phase_v,\Bp_v,\Bc_v)
\]
This variable is essentially meant to represent the current bit-position of the candidate leaders. $\Bs_v$ represents the most significant bit-position among all identifiers, which must be in agreement with variable $\SB_v$ to assess the validity of $\ME_v$. The variables $\Phase_v$ and $\Bp_v$ of $\ME_v$ are the current phase $i$, and the corresponding bit-position revealed by a candidate leader during phase $i$, respectively. The comparison of bits-positions is relevant only if these bits-positions are revealed at the same phase. Hence, we force the system to proceed in phases. 
If, at phase $i$, the bit-position $\rho_v$ of node $v$ is smaller than the bit-position $\rho_u$ of a neighboring node $u$, then node $v$ becomes passive, and $v$ takes $u$ as parent. It is simple to compare two candidate leaders when these candidate leaders are neighbors. Yet, along with the execution of the algorithm, some nodes become passive, and therefore the remaining candidate leaders can be far away, separated by passive nodes. Each passive node is with a positive distance variable $\dis$, and is in a subtree rooted at some candidate leader. Let us now consider  one such subtree $T_v$ rooted at a candidate leader $v$. Whenever $v$ increases its phase from $i$ to $i+1$, and sets the bit-position related to phase $i+1$, all nodes $u$ in $T_v$ must update their variable $\ME_u$ in order to have the same value as $\ME_v$. 

At each phase, trees are merged into larger trees. At the end of phase $i$, all the nodes in a given tree have the same bit-position, and the leaves of the tree inform their parent that the phase is finished. The local variable $\Bc$,  called control, is dedicated to this purpose. Each leaf assigns $1$ to its control variable, and a bottom-up propagation of this control variable eventually reaches the root. In this way, the root learns that the current phase is finished. If the largest identifiers are encoded using $\log n$ bits, each phase results in halving the number of trees, and therefore of candidate leaders. So within at most $\log n$ phases, a single leader remains. To avoid electing an impostor leader, the (unique) leader restarts the election at the first phase. This is repeated forever. If an arbitrary initial configuration induces an impostor leader $\ell$, either $\ell$ has not the most significant bit-position in its identifier or this impostor leader has its most significant bit-position equal to the most significant bit-position of the (real) leader. In the former case, the error is detected by a node with the most significant bit-position. In the latter case,  error is detected by at least one node (the true leader), because there exists at least one phase $i$ where the bit-position of the leader is superior to the bit-position of the impostor.  

The process of leader election and spanning tree construction is slowed down by the hyper-node construction and management. When a node $v$ changes its parents, it also changes its variable $\dB_v$, in order not to  impact the current construction of the tree. The point is that variable $\dB_v$ should be handled with extra care to remain coherent with the tree resulting from merging different trees. To handle this, every candidate leader assigns bits for its children into its variable $\MC$. More precisely, if a root $v$ has not children, then $v$  publishes the bit for its future children with variable distance equal to one. If root $v$ has children with distance variable equal to one, then $v$ publishes the bit for the children $u$ with $\dis_u=2$, and so on, until the distance variable of $v$ becomes $\SB_v$. On the other hand, a node cannot change its parent if its future new parent does not publish the bit corresponding to its future distance variable. When the hyper-node adjacent to the root is constructed, the hyper-node verification process  takes care of the assignment of the bits to the node inside the hyper-node.

\subsection{Formal description of our algorithm $\Algo$}
\label{subsec:Pred}

In the definitions below, the notation $b \equiv P$ is used to define the boolean $b$, which is true if and only if the predicate $P$ is true. 

\subsubsection{Computing the set of children:}
{\footnotesize
$
\begin{array}{rcl}
{\tt Eq\SB}(v) & = & \{ u\in \N_v  \mid  \SB_u=\SB_v \wedge \Bs_u=\Bs_v  \}\\
{\tt EqElec}(v,X)&=&  \{u\in X \mid \ME_u=\ME_v   \}\\
{\tt EqElecP}(v)&\equiv&  \ME_v=\ME_{\p_v} \\
{\tt InfPh}(v)&=&  \{ u\in{\tt Eq\SB}(v)  \mid (\Phase_u=\Phase_v-1) \vee ((\Phase_v=1) \wedge (\Phase_u=\SB_v))  \}\\
{\tt CandC}(v)&=&\multicolumn{1}{l}{ {\tt EqElec}(v, {\tt Eq\SB}(v))\cup {\tt InfPh}(v)}\\
\rowcolor{lightgray}\Child(v) &=&\multicolumn{1}{>{\columncolor{lightgray}}l}{ \{  u\in {\tt CandC}(v) \mid \big((\dis_v<\SB_v) \wedge (\dis_u=\dis_v+1)\big)  \vee \big(\dis_v\in\{0,\SB_v\} \wedge  ( \dis_u=1)\big) \}\hspace*{1,3cm}}\\

\end{array}
$
}

\subsubsection{Hyper-nodes distance verification}
{\footnotesize
$
\begin{array}{rcl}
\VChild(v,M,T)& \equiv & \{ u\in \Child(v)\mid  M=T  \}=\Child(v)\\
\Add_{\tt Nd}(v)&\equiv &(\Add_v=\bot) \wedge (\MC_v=\bot) \wedge (\Child(v)\neq\emptyset)\hspace*{6,2cm}.\\
\Add_{\tt p}  (v)&\equiv &  \big((\dis_v>1) \wedge (\Add_{\p_v}=\bot)\big) \vee (\dis_v=1) \\
\Add_{\tt Ch}  (v)&\equiv & \big ((\dis_v<\SB_v)  \wedge   \VChild(v,\MC,\bot) \big) \vee \big((\dis_v=\SB_v) \wedge  \VChild(v,\MV,\bot) \big)\\
\Add_+(v)& \equiv &   (\dis_v=\SB_v)  \vee  (  \VChild(v,\Add,+) \wedge  \VChild(v,\dB,1))\\
\Add_{ok}(v)&\equiv &   \VChild(v,\Add,ok) \vee  \big (\VChild(v,\Add,+) \wedge  \VChild(v,\dB,0)\big) \\
\rowcolor{lightgray}\TAdd(v)&\equiv &  \Add_{\tt Nd}(v) \wedge  \Add_{\tt p}(v)  \wedge \Add_{\tt Ch}(v) \wedge \big(\Add_+(v) \vee  \Add_{ok}(v)\big)\\

\tt{MCh}(v,k) &  \equiv & \big((\dis_v<\SB_v)  \wedge \VChild(v,\MC,k)\big) \vee \big( (\dis_v=\SB_v) \wedge\VChild(v,\MV,k)\big)\\
{\tt Broad}_{\tt dB}(v)&  \equiv & [(\dis_v=1) \wedge  (\MC_v= \bot)] \vee [(\dis_v>1)\wedge (\MC_v[0]=\dis_v\text{-}1) \wedge \tt{MCh}(v,\MC_v) ]\\
{\tt Broad}_{\tt p_1}(v)& \equiv & (\MC_v=\bot) \wedge (\MC_{\p_v}[0]=1) \wedge \tt{MCh}(v,\emptyset) \\
{\tt Broad}_{\tt p_g}(v)& \equiv & \MC_v \wedge \tt{MCh}(v,\MC_v) \wedge (\MC_{\p_v}[0]=\MC_v[0]+1)\wedge  (\MC_v[0]\neq \dis_v\text{-}1) \\
{\tt Broad}_{\tt p}(v)  &  \equiv & (\dis_v>1) \wedge ({\tt Broad}_{\tt p_1}(v) \vee {\tt Broad}_{\tt p_g}(v))\\
\rowcolor{lightgray}\TBroad(v)  &  \equiv & (\Child(v)\neq\emptyset) \wedge ({\tt Broad}_{\tt dB}(v)\vee   {\tt Broad}_{\tt p}(v))\\

{\tt Vrf}_{\tt Last}(v,M)& \equiv & \neg\Child(v) \wedge (M_{\p_v}[0]=\dis_v+1) \\
{\tt Vrf}_{\tt Start}(v,M)& \equiv & \Child(v) \wedge (M_{\p_v}[0]=\dis_v+1) \wedge  \VChild(v,\MV,\emptyset)\\
{\tt Vrf}_{\tt Broad}(v,M)& \equiv & \Child(v) \wedge (M_{\p_v}[0]>\dis_v) \wedge  \VChild(v,\MV,\MV_v)\\
{\tt Vrf}(v,M)& \equiv &\big({\tt Vrf}_{\tt Last}(v,M) \vee {\tt Vrf}_{\tt Start}(v,M) \vee {\tt Vrf}_{\tt Broad}(v,M)\big)\\
{\tt Vrf}_{\tt 1}(v)& \equiv & (\dis_v=1)  \wedge ( \MC_{\p_v}[0]=\MV_v[0]+1)\wedge {\tt Vrf}(v,\MC) \\
{\tt Vrf}_{\tt g}(v)& \equiv &(\dis_v>1)  \wedge (\MV_{\p_v}[0]=\MV_v[0]+1) \wedge (\ME_{\p_v}=\ME_v) \wedge {\tt Vrf}(v,\MV) \\
{\tt Vrf}_{\tt 1g}(v)& \equiv & \big((\dis_v=1)  \wedge (\MC_{\p_v}[0]=1)\big) \vee \big((\dis_v>1)  \wedge (\MV_{\p_v}[0]=\dis_v) \big)\\
 \rowcolor{lightgray}\TVerif(v) & \equiv &\big [ (\MV_v=\emptyset) \wedge {\tt Vrf}_{\tt 1g}(v)\big] \vee \big [ (\MV_v\neq \emptyset) \wedge \big ({\tt Vrf}_{\tt 1}(v) \vee {\tt Vrf}_{\tt g}(v)\big )\big] \\

{\tt EqElecN}(v,X)&\equiv &  \{  u\in X \mid (\ME_u=\ME_v) \}=\N_v\\
{\tt CleanM}_{\tt V1A}(v)& \equiv &  {\tt EqElecN}(v, {\tt Eq\SB}(v)) \wedge (\MV_v[0]=\SB_v) \wedge [\dis_v=\SB_v \vee (\dis_v\neq\SB_v \wedge \VChild(v,\MV,(\text{-}1,\text{-}1)))]\\
{\tt CleanM}_{\tt V1B}(v)& \equiv &   {\tt EqElecN}(v, {\tt Eq\SB}(v)) \wedge \MV_v=(\SB_v,0) \wedge [\dis_v=\SB_v \vee (\dis_v\neq\SB_v \wedge \VChild(v,\MV,\bot))]\\
{\tt VHC}(v)& \equiv &  \MV_v=(\text{-}1,\text{-}1) \wedge \MV_{\p_v}=(\text{-}1,\text{-}1)\\
{\tt CleanM}_{\tt VA}(v)& \equiv &  {\tt VHC}(v) \wedge [(\dis_v<\SB_v  \wedge \VChild(v,\MV,\emptyset)) \vee (\dis_v=\SB_v \wedge\VChild(v,\MC,\bot))]\\
{\tt CleanM}_{\tt VB}(v)& \equiv &  \MV_v=(\text{-}1,\text{-}1) \wedge (\dis_v=1)  \wedge \VChild(v,\MV,\bot)\\
{\tt CleanM}_{\tt VD}(v)& \equiv &(\dis_v<\SB_v) \wedge (\Child(v)\neq\emptyset) \wedge \VChild(v,\MV,\bot)\\
{\tt CleanM}_{\tt VC}(v)& \equiv &(\MV[0]=\SB_v)\wedge \Big(\big((\dis_v=\SB_v) \wedge (\Child(v)=\emptyset)\big) \vee  {\tt CleanM}_{\tt VD}(v)\Big)\\
{\tt CleanM}_{\tt C}(v)& \equiv &  (\MC_v[0]=\dis_v) \wedge   \tt{MCh}(v,\MC_v)  \wedge  (((\MC_{\p_v}=\bot) \wedge (\dis_v>1)) \vee (\dis_v=1) )\\
\rowcolor{lightgray} \TCleanM(v)& \equiv &    {\tt CleanM}_{\tt V1A}(v) \vee {\tt CleanM}_{\tt V1B}(v) \vee  {\tt CleanM}_{\tt VA}(v) \vee  {\tt CleanM}_{\tt VB}(v) \vee {\tt CleanM}_{\tt C}(v)\\
\Root(v) & \equiv &   (\leader_v=1) \wedge (\dis_v=0) \wedge (\p_v=\emptyset) \wedge  (\SB_v=\Bit(1,\id_v)) \\   
\rowcolor{lightgray}\TStartdB (v)& \equiv &\Root(v) \wedge (\Child(v)\neq \emptyset) \wedge \VChild(v,\MV,\MC_v) \\
\end{array}
$
}

\subsubsection{Leader election and tree construction}

{\footnotesize
$
\begin{array}{lcl}

{\tt Max\SB}(v)& = & \max \{\Bs_u\mid u\in \N_v \}\\ 
{\tt NeigMax\SB}(v)& = &  \{  u\in \N_v \mid\Bs_u={\tt Max\SB}(v)\}\\ 
{\tt EqPh}(v,X)&=&  \{  u\in X \mid (\Phase_u=\Phase_v)   \}\\
{\tt MaxEqP}(v)& = & \max\{ \Bp_u \mid u\in  {\tt EqPh}(v,\N_v) \wedge  \Bp_u>\Bp_v \}\\ 
{\tt NeigMaxEqP}(v)& = & \{ u\in  {\tt MaxEqP}(v) \mid \Bp_u= {\tt MaxEqP}(v)  \}\\   
\Best(v)& = &\left\{
\begin{array}{ll} 
\min \{{\tt port}_u \mid u\in   {\tt NeigMax}\SB(v)\} & \mbox{if }  {\tt Max}\SB(v)>\SB_v \\
\min \{{\tt port}_u \mid  u\in  {\tt NeigMaxEqP}(v)\} & \mbox{otherwise} \\
\end{array} \right.\\ 
\Pass_0(v,x)& \equiv &\big((\dis_x=0) \wedge (\MC_x=(1,0))\big)\\
\Pass_{\tt dB}(v,x)& \equiv &\Pass_0(v,x)\vee \big( (0<\dis_x<\SB_x) \wedge (\MV_x[0]=\dis_x+1)\big)\vee \big((\dis_x=\SB_x) \wedge (\MC_x[0]=1)\big)\\
{\tt EqMEi}(i,X)& = & \{ i\in\{0,...,i\}: \ME_u[i]=\ME_v[i]  \mid u\in X \}\\ 
{\tt EqMEx}(i,x,X)& = & \{ \ME_u[i]=x \mid u\in X\}\\ 
\rowcolor{white}{\tt Wave_B}(v) & \equiv & \big({\tt EqMEi}(2,\Child(v))=\Child(v) \wedge {\tt EqMEx}(3,1,\Child(v))= \Child(v)\big ) \vee \neg\Child(v)\\
\rowcolor{lightgray} \TPass(v)& \equiv &\Best(v)\wedge  \Pass_{\tt dB}(v,\Best(v)) \wedge  {\tt Wave_B}(v)  \\
{\tt SupPh}(v)&=&  \{  u\in{\tt Eq\SB}(v) \mid  (\Phase_u=\Phase_v+1) \vee (\Phase_v=\SB_v \wedge \Phase_u=1) ) \}\\
{\tt Other}(v)&=& \{ u\in{\N_v-\Child(v)-{\tt SupPh}(v)} \mid  \ME_v=\ME_u  \}\\
\rowcolor{lightgray}\TInc(v) & \equiv & \Root(v) \wedge {\tt Wave_B}(v)\wedge  (\Child(v)\cup{\tt SupPh}(v)\cup{\tt Other}(v)=\N_v) \\
\rowcolor{white}{\tt Coh_p}(v)& \equiv & (\Best(v)=\emptyset) \wedge  (\p_v\in\N_v) \wedge (\SB_v=\SB_{\p_v}) \wedge (\dis_v=\dis_{\p_v}-1)\\
\Up_{\tt p}(v) & \equiv & ( \Bc_{\p_v}=0) \wedge   \big( (\Phase_{\p_v}=\Phase_v+1) \vee ((\Phase_v=\SB_v) \wedge( \Phase_{\p_v}=1))  \big)\\
\Up_{\tt E}(v) & \equiv & (\Bc_v=0) \wedge ( \Bc_{\p_v}=0)\\
\Up_{\tt Back}(v) & \equiv & \Up_{\tt E}(v) \wedge ({\tt EqMEi}(v,2,\N_v)=\N_v) \wedge ({\tt EqMEx}(2,\Child(v),1)=\Child(v))\\
\rowcolor{lightgray}\TUp(v) & \equiv &{\tt Coh_p}(v) \wedge (\Up_{\tt p}(v)   \wedge {\tt Wave_B}(v)) \vee  \Up_{\tt Back}(v) \\
\end{array}
$
}

\subsubsection{Reset and Error detection}
\label{Pred:Errors}
{\footnotesize
$
\begin{array}{rcl}
\tt{MReset}(v) & \equiv &    (\dB_v=0) \wedge (\Add_v=\bot)  \wedge (\MV_v=\bot)\\                    
\tt{VReset}(v) & \equiv & (\leader_v=0)  \wedge  (\dis_v=\text{-}1)  \wedge  (\SB_v=\text{-}1) \wedge (\ME_v=(\text{-}1,0,\text{-}1,\text{-}1)) \\     
\Nd{\tt Reset}(v)  & \equiv & \tt{MReset}(v) \wedge (\MC_v=\bot)  \wedge \tt{VReset}(v)\\

\Nd{\tt Start}(v) & \equiv &  \tt{MReset}(v) \wedge (\MC_v\neq\bot) \wedge \Root(v) \wedge (\ME_v=(\Bit(1,\id_v),1,\Bit(1,\id_v),\{0,1\}))\\
\Neig{\tt Reset}(v) & = & \{u\in \N_v \mid \Nd{\tt Reset}(u)   \}\\
\Neig{\tt Start}(v)& = & \{ u\in \N_v \mid \Nd{\tt Start}(u) \}\\
\rowcolor{lightgray} \TReset (v) & \equiv & \neg \Nd{\tt Reset}(v) \wedge \neg \Nd {\tt Start}(v) \wedge |\Neig{\tt Reset}(v)|>0 \\     
\rowcolor{lightgray}\TStart(v)&   \equiv &\Nd{\tt Reset}(v) \wedge  \big (\{\Neig{\tt Reset}(v)\cup  \Neig{\tt Start}(v)\}=\N_v \big) \\
\inNd(v)  & \equiv &   (\leader_v=0) \wedge (\dis_v>0) \wedge  {\tt Coh_p}(v) \wedge  (\SB_v\geq\Bit(1,\id_v)) \wedge (\Bs_v=\SB_v) \\

\E_{\tt d}(v)  & \equiv &   (\dis_v>0) \wedge   \wedge ({\tt Best(v)}=\emptyset) \neg\big( (\dis_{\p_v}<\SB_v \wedge \dis_v\neq \dis_{\p_v}+1)\vee (\dis_{\p_v}\in\{0,\SB_v\} \wedge \dis_v\neq 1)\big)\\
\rowcolor{LLgray}\E_{\tt Nd}(v)& \equiv &  ( \neg\Root(v)  \wedge \neg\inNd(v) \wedge \neg \Nd{\tt Reset}(v)) \vee (\inNd(v) \wedge \E_{\tt d}(v)) \\        
{\tt NgPh}(v)&=&\{  \Phase_u \mid  u\in \N_v\}\\
\E_{\tt PhMinB}(v)& \equiv & \big((\Phase_v=\SB_v) \wedge (\min\{{\tt NgPh}(v)\}\not \in\{1,\SB_v\text{-}1,\SB_v\})\big))\\
\E_{\tt PhMinG}(v)& \equiv & \big((\Phase_v\neq\SB_v) \wedge (\Phase_v-\min \{{\tt NgPh}(v)\}>1)\big)\\
\E_{\tt PhMax1}(v)&\equiv & \big((\Phase_v=1) \wedge (\max\{{\tt NgPh}(v)\}\not \in\{2,\SB_v\})\big)\\
\E_{\tt PhMaxG}(v)&\equiv & \big((\Phase_v\neq1) \wedge (\max \{{\tt NgPh}(v)\} - \Phase_v>1)\big) \\
\rowcolor{LLgray}\E_{\tt Phase}(v)&\equiv & \big(\E_{\tt PhMinB}(v) \vee \E_{\tt PhMinG}(v)  \vee \E_{\tt PhMax1}(v) \vee \E_{\tt PhMaxG}(v)\big)\\
\rowcolor{LLgray}\E_{\tt Bp}(v)& \equiv &   (\dis_v>0) \wedge \{\{   u\in {\tt EqPh}(v,{\tt Eq\SB}(v))\mid \Bp_u=\Bp_v \}\cup {\tt SupPh}(v) \cup \{\Best(v)\}\}=\emptyset\\
\rowcolor{LLgray}\E_{\tt Control}(v)& \equiv & ((\Bc_v=1) \wedge  \neg {\tt Wave_B}(v)) \vee  ((\Bc_v=0) \wedge (\Bc_{\p_v}=1))\\
\end{array}
$
}

\hspace*{-0,6cm}
{\footnotesize
$
\begin{array}{rcl}
\E_{\tt MRoot}(v)& \equiv & (\dis_v=0) \wedge \big((\Add_v\neq \bot) \vee (\MV_v\neq \bot)\big)\\
\E_{\tt MAdd}(v)& \equiv & (0<\dis_v<\SB_v) \wedge (\Add_v\neq \bot) \wedge \VChild(v,\Add,\bot)\hspace*{6cm}.\\
\E_{\tt PL}(v)& \equiv & (\dis_v>1) \wedge \big((\MC_v[0]>\dis_v) \vee ( (\Best(v)=\emptyset) \wedge (\MC_{\p_v}=\bot)\wedge (\MC_v[0]<\dis_{\p_v})\big)\\
\E_{\tt HCch}(v)& \equiv &  (\dis_v<\SB_v) \wedge (\MV_v=\bot) \wedge \VChild(v,\MV,\neg\bot)\\
\E_{\tt HCp}(v)& \equiv & (\dis_v>0) \wedge (\MV_v\neq\bot) \wedge (\Best(v)=\emptyset) \wedge (\MV_{\p_v}[0]<\MV_v[0])\\
\E_{\tt HC}(v)& \equiv & (0<\dis_v<\SB_v) \wedge \big( (\MV_v[0]<\dis_v) \vee \E_{\tt HCch}(v) \vee \E_{\tt HCp}(v)  \big)\\
\rowcolor{LLgray}\E_{\tt Mem}(v)& \equiv & \E_{\tt MRoot}(v) \vee \E_{\tt MAdd}(v) \vee \E_{\tt PL}(v) \vee \E_{\tt HC}(v) \\
\rowcolor{lightgray}\E_{\tt T}(v)& \equiv &  (\Neig{\tt Reset(v)}=\emptyset) \wedge \big(\E_{\tt Nd}(v) \vee \E_{\tt Phase}(v) \vee \E_{\tt Bp}(v) \vee \E_{\tt Control}(v) \vee \E_{\tt Mem}(v)\big)\\
\E_{\tt Add }(v)& \equiv &  (\dis_v=1) \wedge (\dB=1)\wedge \VChild(v,\Add,+) \wedge  \VChild(v,\dB,1)\\
\E_{\tt  H1}(v)& \equiv &(\dis_v=1) \wedge (\MC_{\p_v}[0]=\dis_v) \wedge (\MC_{\p_v}[1]\neq \dB_v)\\
\E_{\tt  Hg}(v)& \equiv &(\dis_v>1)  \wedge (\MV_{\p_v}[0]=\dis_v) \wedge (\MV_{\p_v}[1]\neq \dB_v)\\
\rowcolor{LLgray}\E_{\tt  Hyper}(v)& \equiv &  \E_{\tt Add }(v) \vee \big((\ME_v=\ME_{\p_v}) \wedge  (\E_{\tt  H1}(v) \vee \E_{\tt  Hg}(v))\big)\\
\rowcolor{LLgray}\E_{\tt Elec}(v)& \equiv & (\dis_v=0) \wedge  (\SB_v=\Bit(1,\id_v)) \wedge (\Bp_v<\Bit(\Phase_v,\id_v))\\
\rowcolor{lightgray} \Error(v)&   \equiv & \E_{\tt T}(v) \vee \E_{\tt  Hyper}(v)  \vee \E_{\tt Elec}(v)  \\

\end{array}
$
}

\subsubsection{Commands}

\subsubsection*{Commands for hyper-node distance verification:}

\begin{minipage}{8cm}
{\footnotesize
\vspace*{-1,1cm}
\begin{description}
\item $\Verification(v)$:\\
$(\dis_v=1) \wedge (V_{\p\dis}(v) \vee V_{\p 1}(v))\rightarrow \MV_v=\MC_{\p_v}$\\
$(\dis_v>1) \wedge (V_{\p\dis}(v) \vee V_{\p \bar{1}}(v)) \rightarrow \MV_v=\MV_{\p_v}$
\item $\CleanM(v)$:\\
${\tt CleanM}_{\tt V1A}(v) \vee {\tt CleanM}_{\tt V1B}(v) \rightarrow \MV_v=(\text{-}1,\text{-}1)$\\
${\tt CleanM}_{\tt VA}(v) \vee  {\tt CleanM}_{\tt VB}(v) \rightarrow \MV_v=\bot$\\
${\tt CleanM}_{C}(v) \rightarrow \MC_v=\bot$\\
\end{description}
}
\end{minipage}
\begin{minipage}{8cm}
{\footnotesize
\begin{description}
\item $\BinaryAd(v)$:\\
$\Add_{+}(v) \rightarrow \Add_v=+$\\
$\Add_{ok}(v) \rightarrow  \Add_v=ok$\\
\item $\Broadcast(v)$:\\
${\tt Broad}_{\tt dB}(v)\wedge (\Add_v=ok) \rightarrow$\\
\hspace*{1.5cm} $ \MC_v=(\dis_v,\dB_v);\Add_v=\bot$\\
${\tt Broad}_{\tt dB}(v)\wedge (\Add_v=+) \rightarrow$\\
\hspace*{1.5cm} $ \MC_v=(\dis_v,\overline{\dB}_v);\Add_v=\bot$\\
${\tt Broad}_{\tt p}(v)   \rightarrow \MC_v=\MC_{\p_v}$\\
\end{description}
}
\end{minipage}

\subsubsection*{Commands for the leader election and tree construction:}

\begin{minipage}{8cm}
{\footnotesize
\begin{description}
\item $\Inc(v) :$\\
$T \equiv (\Bp_v=\text{-}1) \wedge (\Phase_v=\SB_v+1)$\\
$T\rightarrow i:=1$; $\neg T\rightarrow i:=\Phase_v+1;$\\
$ \ME_v:= (\Bit(1,v),i,\Bit(i,v),0);$
\item $\StartdB(v)$:\\
$(\MC_v[0]=\SB_v) \rightarrow \MC_v:=(1,0)$\\
$(\MC_v[0]<\SB_v) \rightarrow \MC_v:=(\MC_v[0]+1,0)$
\item \Update(v):\\
${\tt Wave_B}(v)  \wedge \Up_{\tt p}(v) \rightarrow \ME_v=\ME_{\p_v}$\\
$\Up_{\tt Back}(v)\rightarrow \Bc_v=1$\\
\end{description}
}
\hspace*{1cm}
\end{minipage}
\vspace*{-1,1cm}
\begin{minipage}{8,5cm}
{\footnotesize
\vspace*{-2,3cm}
\begin{description}
\item $\Passive(v):$\\
$\leader_v:=0$; $\p_v:=\Best(v)$; $\dB_v:=\MV_{\p_v}[0]$;\\
$(\dis_{\p_v}=\SB_{\p_v})\rightarrow \dis_v:=1$\\
$(\dis_{\p_v}<\SB_{\p_v})\rightarrow \dis_v:=\dis_{\p_v}+1$\\
$\SB_v:=\SB_{\p_v}$; $\ME_v:=\ME_{\p_v}$\\
$\Add_v:=\bot$; $\MV_v:=\bot$; $\MC_v:=\bot$;\\
\end{description}
}
\end{minipage}

\subsubsection*{Commands activated after error detection:}

\begin{minipage}{8,5cm}
{\footnotesize
\begin{description}
\item $\Reset(v):$\\
\hspace*{-0.5cm}  $\leader_v:=0$; $\p_v:=\emptyset$; $\dis_v:=\text{-}1$; $\dB_v:=0$;\\
\hspace*{-0.5cm} $\SB_v:=\text{-}1$; $\ME_v:=(\text{-}1,0,\text{-}1,\text{-}1)$;\\
\hspace*{-0.5cm} $\Add_v:=\bot$; $\MV_v:=\bot$; $\MC_v:=\bot$;\\
\end{description}
}
\end{minipage}
\begin{minipage}{8,5cm}
{\footnotesize
\begin{description}
\item $\Start(v):$\\
\hspace*{-0.5cm}  $\leader_v:=1$; $\dis_v:=0$;\\
\hspace*{-0.5cm}  $\SB_v=\Bit(1,v)$; $\ME_v=(\Bit(1,v),1,\Bit(1,v),0)$; \\
\hspace*{-0.5cm}  $\MC_v=(1,0)$;\\
\end{description}
}
\end{minipage}

\newpage

\section{Correctness}
\label{sec:Correctness}

In this section, we formally prove the correctness of our Algorithm. 

\begin{theorem}
\label{theo:election}
Algorithm $\Algo$ solves the leader election problem in a self-stabi\-li\-zing manner for the $n$-node ring, in the state model, with a distributed weakly-fair scheduler. Moreover, if the $n$ node identities are in the range $[1,n^c]$ for some $c\geq 1$, then Algorithm $\Algo$ uses $O(\log \log n)$ bits of memory per node, and stabilizes in $O(n \log^2 n)$ rounds. 
\end{theorem}

The main difficulty for proving this theorem is to prove that $\Algo$ can detect any cycle generated by the parenthood relation in the initial configuration, and can, whenever a cycle is detected, remove this cycle. Let $\Gamma$ be the set of all possible configurations of the ring, under the set of variables described before in the paper. First, we prove that Algorithm $\Algo$  detects the presence of ``trivial'' errors, that is, inconsistencies between neighbors. Second, we prove that, after correcting all the trivial errors (possibly using a reset), $\Algo$ converges and maintains configurations free of  trivial errors. The set of configurations free of  trivial errors is denoted by  $\Gamma_{\tt TEF}$ where TEF stands for ``Trivial Error Free''. From now on, we assume only configurations from $\Gamma_{\tt TEF}$.

The core of the proof regarding proper cycle detection is based on proving the correctness of the hyper-node distance verification process. This verification process is the most technical part of the algorithm, and proving its correctness is actually the main challenge in the way of establishing Theorem~\ref{theo:election}. This is achieved by using proofs based on invariance arguments. 

Once the correctness of the hyper-node distance verification process has been proved, we establish the convergence of Algorithm $\Algo$ from an arbitrary configuration in $\Gamma_{\tt TEF}$ to a configuration without cycles, and where all hyper-node distances are correct. The set of configurations without cycles is denoted by $\Gamma_{\tt CF}$ (where CF stands for ``Cycle Free'').
We prove that a configuration is in $\Gamma_{\tt CF}$ if and only if all hyper-node distances are correct.
 Once we can restrict ourselves to configurations in  $\Gamma_{\tt CF}$, we prove the correctness of our mechanisms detecting and removing impostor leaders. We denote by  $\Gamma_{\tt IEF}$ (where IEF stands for ``Impostor leader Error Free'') the set of configurations with no impostors. Finally, assuming a configuration in $\Gamma_{\tt IEF}$, we prove that the system reaches and maintains a configuration with exactly one leader, equal to the node with maximum identity. Moreover, we prove that the structure induced by the parenthood relation is a tree rooted at the leader, and spanning all nodes. We denote by $\Gamma_{\tt LE}$ the set of configurations where the unique leader is the node with maximum identity. In other words, we prove that $\Algo$  is self-stabilizing for $\Gamma_{\tt LE}$. 

In the statements of the lemmas below, we define predicates on configurations, these predicates are used as attractors toward a legitimate configuration (i.e., a configuration with unique  leader). To establish convergence toward attractors, we use potential functions~\cite{T09bc}, that is, functions that map configurations to non-negative integers, and that strictly decrease after each algorithm command is executed. 

Before starting the proofs, let us first define the Predicate $\Gamma_{\tt LE}$ (\emph{Leader Election}) that serves as the definition for legitimate configurations.
Let $L: \Gamma\rightarrow\mathbb{N}$ be the function defined by
$$L(\gamma)=\sum_{v \in V}\leader_v~.$$
A configuration $\gamma$ is legitimate for the leader election specification, i.e., satisfies $\Gamma_{\tt LE}$, if and only if $L(\gamma)=1$.   That is, 
\[
\Gamma_{\tt LE}=\{\gamma\in\Gamma : L(\gamma)=1\}.
\]
Now, for the purpose of the proof, let us define Predicate $\Gamma_{\tt TEF}$ (\emph{Trivial Error Free}).
Let $\psi: \Gamma\times V \rightarrow \mathbb{N}$ be the function defined by:
$$\psi(\gamma,v)=
\left\{
  \begin{array}{ll}
    1 &\text{if } \E_{\tt T}(v) \; \mbox{is true}\\
    0 &\text{otherwise} \\
  \end{array}
\right.
$$
Let $\Psi: \Gamma \rightarrow\mathbb{N}$ be the function defined by:
$$\Psi(\gamma)=\sum_{v\in V}\psi(\gamma,v)~.$$ 
Note that all nodes are legitimate with respect to $\E_{\tt T}(v)$ if and only if  $\Psi(\gamma)=0$. We define 
\[
\Gamma_{\tt TEF}=\{\gamma\in\Gamma : \Psi(\gamma)=0\}.
\]

\begin{lemma}
\label{lem:TEF}
$true \triangleright \Gamma_{\tt TEF}$ in one round.
\end{lemma}

\begin{proof}
The predicate $\E_{\tt T}(v)$ is composed of two types of sub-predicates: the ones independent of the local variables of the neighbors, called \emph{independent} sub-predicates, and the ones based on comparisons with the local variables of the neighbors, called \emph{dependent} sub-predicates.  (See predicate in~\ref{Pred:Errors}). The independent sub-predicates are $\Root$, $\inNd$, and $\E_{\tt MRoot}$. All other sub-predicates of $\E_{\tt T}$ are dependents. If an independent sub-predicate is true, then  $\E_{\tt T}(v)$ is true,  and the execution of any command by a neighbor of $v$ has no influence on the predicate $\E_{\tt T}(v)$, which remains true. (Note that if an independent sub-predicate of $v$ is true, then $v$ may have all its neighbors $w$ with  $\E_{\tt T}(w)$ equal to false).   Let  us now consider dependent sub-predicates. Let $v$ a node with a dependent sub-predicate $sp(v)=true$. Then $\E_{\tt T}(v)=true$, and there is at least one neighbor $u$ of $v$ such that some local variable(s) of $u$ is not consistent with some variable(s) of $v$. In this case, since the sub-predicate $sp$ is dependent, we also have $sp(u)=true$, as, if one variable of $v$ is not consistent with the corresponding variable of $u$, then the converse is also true. Therefore, $sp(u)$ is true, and thus $\E_{\tt T}(u)$ is true as well. As a consequence, nodes $u$ and $v$ can only apply the command $\Reset$.  Now let us consider the neighbor $w$ of $v$ distinct from $u$, and the  neighbor $y$ of $u$ distinct from $v$. The predicates $\E_{\tt T}(w)$ and $\E_{\tt T}(y)$ may be equal to false, in which case $w$ and $y$ can execute a command different from $\Reset$. Nevertheless, the execution of one command by node $w$ or $y$, or both, cannot modify the variables of $v$ and $u$. Therefore, the inconsistency between the local variables of $v$ and $u$ remains until one of them apply the reset command. 

Let us denote by $X$ the set of nodes $v$ such that $\E_{\tt T}(v)=true$ (see predicate in~\ref{Pred:Errors}) and  $\E_{\tt T}(y)=false$ for every $y\in \N_v$. Moreover, let us denote by $Y$ the set of nodes $v$ such that $\E_{\tt T}(v)=true$  and there exists $y\in \N_v$ such that $\E_{\tt T}(y)=true$. Let $\gamma_0$ denote the initial configuration of the system. Since the scheduler is weakly fair, all nodes in $X$ are scheduled for execution at round $1$. Moreover all nodes in $Y$ are also scheduled at round $1$. According to rule $\R_{\tt Error}$ (see algorithm in Figure~\ref{fig:Algo}), every node $v\in X\cup Y$ executes $\Reset(v)$, and thus $\E_{\tt T}(v)$ becomes false for all these nodes as well as for all their neighbors (indeed, if one neighbor of $v$ is reset, then $\E_{\tt T}(v)$ is false). Therefore, the configuration $\gamma_1$ at time~1 satisfies 
$$\Psi(\gamma_1)=0.$$
Thus,  starting from any arbitrary initial configuration $\gamma_0$, the system reaches a  configuration $\gamma_1 \in \Gamma_{\tt TEF}$ in just one round. 
\end{proof}

\begin{lemma}
\label{lem:TEF-close}
$\Gamma_{\tt TEF}$ is closed.
\end{lemma}

\begin{proof}
We  prove that, starting from a configuration where $\Gamma_{\tt TEF}$ holds, algorithm $\Algo$ preserves that $\E_{\tt Nd}(v)$,  $\E_{\tt Phase}$, $\E_{\tt Bp}$, $\E_{\tt Control}$, and $\E_{\tt Mem}$ all remain false. Let us consider these predicates one by one, starting with $\E_{\tt Nd}(v)$. We show that, for all $v\in V$, algorithm $\Algo$ keeps predicate $\E_{\tt Nd}(v)$ to false starting from a configuration where $\Gamma_{\tt TEF}$. In configuration $\Gamma_{\tt TEF}$, each node $v$ has $\E_{\tt Nd}(v)=false$, that is, each node $v$ must satisfy one of the three predicates $\Root(v)$, $\inNd(v)$, and $\NdReset(v)$. Moreover, if $v$ satisfies $\inNd(v)=true$, then $v$ must have a correct distance (see the definition of $\E_{\tt d}(v)$). We consider the different commands that can be executed by the algorithm. 

\begin{itemize}

\item If a node $v$ executes the command $\Reset$ due to a hyper-node distance error, i.e., $\E_{\tt  Hyper}(v)=true$,  or to an election error, i.e., $\E_{\tt Elec}(v)=true$, then the predicate  $\NdReset(v)$ becomes true, and, as a consequence, $\E_{\tt Nd}(v)$ remains false. Let node $u$ be a neighbor of $v$. Predicate $\E_{\tt T}(u)$ remains false because $u$ has at least one neighbor that applied the reset command.

\item A reset node $v$, i.e., a node $v$ such that $\dis_v=\text{-}1$ and $\NdReset(v)=true$, can only execute the command $\Start(v)$. When this command has been executed, the predicate $\Root(v)$ becomes true, and $\dis_v$ becomes null.

\item For the command $\Inc$ and $\StartdB$, we note that the predicate $\Root(v)$ remains true if $\dis_v=0$ since no rule modifies the variables used in $\Root(v)$.

\item A candidate node $v$, i.e., a node $v$ such that $\dis_v=0$ and $\Root(v)=true$, can only execute the command $\Passive(v)$. This command is activated if and only if  $v$ has a ``better'' neighbor $u=\Best(v)$. In this case, $v$ adjusts $\dis_v$ with respect to  $\dis_u$. As a consequence, we get $\E_{\tt d}(v)=false$. Moreover, when the command $\Passive(v)$ has been executed, the predicate $\Root(v)$ becomes false, and the predicate $\inNd(v)$ becomes true. Therefore, $\E_{\tt Nd}(v)$ remains false.

\item The predicate $\inNd(v)$ remains true because $v$ executes the commands corresponding to $\R_{\tt Passive}$, $\R_{\tt Update}$, or $\R_{\tt HyperNd}$. We examine each of these different commands. 

\begin{itemize}

\item The commands corresponding to the rule $\R_{\tt HyperNd}$ do not change the variables used by  $\inNd(v)$ and  $\E_{\tt d}(v)$. Hence  $\E_{\tt Nd}(v)$  remains false.

\item Note that for any two node $u$ and $v$ where $u$ is a child of $v$, node $u$ can only execute $\Update(u)$ after $v$ has executed  $\Update(v)$. The command $\Update(v)$ keeps  the predicate $\inNd(v)=true$  because this command does not change the variables $\leader_v$, $\p_v$, $\dis_v$, and $\SB_v$. Moreover, a node $v$ updates its variables according to the predicate $\TUp(v)$ that satisfies $\SB_v=\SB_{\p_v}$. So, after $\Update(v)$ has been performed, we have $\Bs_v=\SB_v$. The command $\Update(v)$ does not modify the parent $w$ of node $v$. Let $u$ be the neighbor of $v$ different from $w$. If $u$ and $w$ do not execute the command $\Passive$, then $\Best(v)=\emptyset$ and $\E_{\tt d}(v)$ remains false. Note that the variable $\Bc_v$ is not used in predicates $\inNd(v)$ and $\E_{\tt d}(v)$. If the parent $w$ of $v$ executes the command $\Update(w)$, then $\dis_w$ does not change, and, as a consequence, $\E_{\tt d}(v)$ remains false. 

\item  Let $u$ be the neighbor of $v$ such that  $u=\Best(v)$.  If $v$ executes $\Passive(v)$, then $\dis_v=\dis_u+1$. Therefore, the new distance of $v$ is  correct by definition of  $\E_{\tt d}$. Moreover since $u=\Best(v)$, we have $\SB_u \geq \Bit(1,Id_v)$. If the parent $u$ of $v$ executes $\Passive(u)$ and updates $\dis_v$, then $u=\Best(v)$, and thus $\Best(v)\neq \emptyset$. As a consequence, ${\inNd}(v)$ remains true and $\E_{\tt d}$ remains  false.
\end{itemize}
\end{itemize}

We now move on to predicate $\E_{\tt Phase}$. We show that, for every $v\in V$, algorithm $\Algo$ keeps predicate $\E_{\tt Phase}(v)=false$. An error is detected if and only if the difference between the phases of two nodes is larger than one. As before, we consider the different commands that can be executed, one by one, according to the value of $\dis_v$. 

\begin{itemize}

\item Let $v$ be a node with $\dis_v=-1$ and $\Phase_v=0$. In this case, node $v$ can only execute the command $\Start$, which sets $v$'s phase to zero. In accordance with the rule $\R _{\tt Start}$, each neighbor $u$ of $v$ has either reset (i.e., $\Phase_u=0$) or restarted (i.e., $\Phase_u=1$), which keeps $\E_{\tt Phase}(v)=false$.

\item Let $v$ be a node with $\dis_v=0$ and $\Phase_v=i$, for some $i\geq 0$. In this case, node $v$ can only execute the command $\Passive(v)$ or the command $\Inc(v)$. In the latter case, the phase $i$ of $v$ can only be increased if $\Phase_u=i$
of $i+1$ for every neighbor $u$ of $v$. Therefore, the command $\Inc(v)$ maintains $\E_{\tt Phase}(v)= false$. We now consider command $\Passive(v)$. When a $v$ node chooses a neighbor $u=\Best(v)$ with $\SB_u>\SB_v$, the phase of $u$ is equal to~1. Let us consider  the subtree $T_r$ rooted in $r$, and containing  the node $v$, with possibly $v=r$. The node $r$ increases its phase $i$ if and only if all its descendants are in phase~$i$, and the neighbors of the leaves of $T_r$ are in phase $i$ or $i+1$ (see the predicates ${\tt Wave_{B}}$ and $\Up_{\tt Back}$ ). As a consequence, the difference of phases between $u=\Best(v)$ and $v$ is at most one. This keeps  $\E_{\tt Phase}(v)=false$. If $\SB_u=\SB_v$, then  node $u=\Best(v)$  if its phase is the same as $v$, and its bit-position is greater than the bit-position of $v$. Therefore, $\E_{\tt Phase}(v)$ remains  false.

\item Let $v$ be a passive node with $\dis_v>0$ and $\Phase_v=i$, for some $i\geq 0$. Only two commands can change the phase $\Phase_v$: the command $\Passive(v)$, and the command  $\Update(v)$. The command $\Passive(v)$ has already be considered in the previous item. Let us thus now consider the command $\Update(v)$, and let $r$ be the root of the subtree $T_r$ containing $v$. When $r$ increases its phase, all nodes in $T_r$ have their phase equal to $i$. Let $L_r$ be the set of  leaves of $T_r$. All nodes $u$ such that $u\notin T_r$ with $u\in\N_x$ for some $x\in L_r$ have their phase equal either to $i$ or to $i+1$ (this is ensured by the control variable $\Bc_v$). As a consequence, if $v$ moves to phase $i+1$, then its parent is at phase $i+1$. If $v\not\in L_r$, then its children are at  phase $i$. If $v\in L_r$ then the neighbors of $v$ not in $T_r$ remain at phase $i$ or at phase $i+1$. Therefore, $\E_{\tt Phase}(v)$ remains false.

\end{itemize}

We now move on to predicate $\E_{\tt Bp}$. We show that, for every $v\in V$, algorithm $\Algo$ keeps predicate $\E_{\tt Bp}(v)=false$. The predicate $\E_{\tt Bp}(v)$ is satisfied if and only if $\dis_v>0$. A node $v$ that is passive at phase $i$ satisfies  $\E_{\tt Bp}(v)=false$ if and only if $v$ has not ``better'' neighbor  (i.e., $\Best(v)=\emptyset$), none of its neighbors are at phase $i+1$ (i.e., ${\tt SupPh}(v)=\emptyset$), and at least one of its neighbors is at the same phase $i$ with the  same bit position as~$v$. Note that, the commands corresponding to the rule $\R_{\tt HyperNd}$ do not change the variables used by the predicate  $\E_{\tt Bp}(v)$. As a consequence,  $\E_{\tt Bp}(v)$  remains false. The only commands that may change $\E_{\tt Bp}$ are commands $\Update$, and $\Passive$. Let us again consider the subtree $T_r$ rooted in $r$ containing $v$, and let $i$ be the current phase of the nodes in $T_r$. Note that, in this case,  $\E_{\tt Bp}(v)$ is false because the parent of $v$ has a same phase as $v$, and the same bit-position as $v$.  If $r$ increases its phase, then the ancestors  of $v$ update their phase top-down. Hence, when the parent $u$ of $v$ updates its phase, the predicate ${\tt SupPh}(v)$ becomes not empty. If $r$ becomes passive, then ancestors of $v$ have a ``better'' neighbor, and thus they set their variables according to the variables of their ``better'' neighbors. When the parent $u$ of $v$ executes the command $\Passive(u)$, the predicate $\Best(v)$ becomes not empty. Hence, in all cases, $\E_{\tt Bp}(v)$  remains false.

We now turn our attention to predicate $\E_{\tt Control}$. We prove that, for every $v\in V$, algorithm $\Algo$ keeps predicate $\E_{\tt Control}(v)$ to false. The boolean $\Bc_v$ can only be modified by the command $\Update$ because of predicate $\TUp$. $\Bc_v$ is used to control the updates of a subtree.  Let us consider $T$ the subtree containing $v$, and let $r$ be the root of $T$. The updating process starts at node $r$, and all the descendants of $r$ set their variables $\Bc=0$ top-down. The process ends at the leaves of the subtree $T$ that launch the acknowledgement $\Bc=1$ propagated bottom-up to the root. Hence, when a node $v$ satisfies $\Bc_v=1$, all its descendants also satisfy  $\Bc=1$, and, conversely, when a node $v$ satisfies $\Bc_v=0$, all its ancestors satisfy $\Bc=0$. Therefore, $\E_{\tt Control}(v)$ is kept to false.
 
Finally, we consider the predicate $\E_{\tt Mem}$, and prove that, for every  $v\in V$, algorithm $\Algo$ keeps predicate $\E_{\tt Mem}$ to false. The predicate $\E_{\tt Mem}$ is satisfied if and only if $\dis_v\geq 0$. The commands $\Inc(v)$, $\StartdB(v)$, and $\Update(v)$   do not change the variables $\Add_v$ and $\MV_v$. Thus  $\E_{\tt Mem}(v)$ remains false. The command $\Passive$ is executed during the construction of the spanning tree, when the hyper-nodes are set. Note that a hyper-node does not start the binary addition before it is properly set. When a node $v$ becomes passive with $\Best(v)=u$,  the variables $\Add_v$ and $\MV_v$ are set to $\bot$. As a consequence,  the predicate $\E_{\tt MAdd}(v)$ is false. Assume that $v$ joins a hyper-node or that $v$ is the first node of a hyper-node. In both cases, $v$ has no children. Thus, $\E_{\tt PL}(v)$ and $\E_{\tt HC}(v)$ remain false. A neighbor $w$ of $v$ with $\p_w=v$ has now $\Best(w)=v$. As a consequence the three predicates $\E_{\tt MAdd}(w)$, $\E_{\tt PL}(w)$ and $\E_{\tt HC}(w)$ remain false. We complete the proof by analyzing all commands corresponding to the rule $\R_{\tt HyperNd}$. 

\begin{itemize}

\item Let us consider the predicate $\E_{\tt MAdd}(v)$ at some node $v$. In this predicate, $\Add_v$ is compared to $\Add_u$ for every child $u$ of $v$. Only two commands change $\Add_v$:   $\BinaryAd$ and $\Broadcast$. The command $\BinaryAd$ assigns to $\Add$ the values $+$ or $ok$ in a bottom-up manner, starting at all nodes $v$ satisfying $\dis_v=\SB_v$. The command $\Broadcast$ sets $\Add$ to $\bot$ in top-down manner, starting at all nodes $v$ with $\dis_v=1$. As a consequence, $\E_{\tt MAdd}(v)$ remains false.

\item Let us consider the predicate $\E_{\tt PL}(v)$. In this predicate, $\MC_v$ is compared to $\MC_{\p_v}$.  Only two commands can change $\MC_v$:  $\Broadcast$ and $\CleanM$. The command $\Broadcast$ sets $\MC$  in a top-down manner starting at every node $v$ with $\dis_v=1$. The command $\CleanM$ sets  $\MC$ to $\bot$ whenever a child $u$ of a node $v$ has  $\MC_u=\MC_v$. Let us consider a hyper-node $X=\{x_1,\dots,x_k\}$. At the end of the binary addition, all the variables $\MC$ of $X$ are equal to $\bot$. The command $\Broadcast(x_1)$ sets $\MC_{x_1}[0]=\dis_{x_1}$, which keeps $\E_{\tt PL}(x_1)$ at false because $\dis_{x_1}=1$. The predicate ${\tt Broad_p}$ in the command $\Broadcast(x_2)$ sets $\MC_{x_2}=\MC_{x_1}$ that maintains $\E_{\tt PL}(x_2)$ false because  $\MC_{x_2}[0]=\dis_{x_1}$. The node $x_1$ can execute the command $\CleanM$ that sets $\MC_{x_1}=\bot$. In this case, $\E_{\tt PL}(x_2)$ remains false. The node $x_3$ executes $\Broadcast(x_3)$ and sets $\MC_{x_3}[0]=\dis_{x_1}$. Then,  $x_2$ sets $\MC_{x_2}[0]=\dis_{x_2}$ and $\E_{\tt PL}(x_2)$ remains false. This process is repeated in a top-down manner for all the nodes in $X$. Hence, the algorithm \Algo\/ keeps $\E_{\tt PL}(v)= false$ for every node $v$.

\item Finally, let us consider the predicate $\E_{\tt HC}(v)$. In this predicate,  $\MV_v$ is compared to $\MV_{\p_v}$ and $\MV_u$ for every child $u$ of $v$. The hyper-node distance is checked in a top-down manner. For that purpose, two commands are used: $\Verification$ and $\CleanM$.  Let us consider a hyper-node $X=\{x_1,\dots,x_k\}$. At the beginning  of the hyper-node distance verification, all the variables $\MV$ of $X$ are equal to $\bot$. According to the predicate $\TVerif$, every node $x_i$, $1< i \leq k$, sets $\MV_{x_i}[0]=\MV_{x_{i-1}}[0]$ if and only if either $\MV_{\p_{x_i}}[0]=\dis_{x_i}$, or  $\MV{\p_{x_i}}[0]=\MV_{x_i}[0]+1$. In both cases $\E_{\tt HC}(v)$ remains false.
\end{itemize}
\end{proof}

We now define Predicate $\Gamma_{\tt CF}$ (for \emph{Cycle Free}), which ensures that no cycles induced by the $\p$ variable remain in the network. Let $\lambda:\Gamma\times V\rightarrow \mathbb{N}$ be the function defined by:
$$\lambda(\gamma,v)=|\B_{P_X} - \B_X - 1|$$
Where $X$ is an hyper-node, and $P_X$ is the hyper-node parent of $X$ (recall that $\B_X$ is an integer whose binary representation is $\dB_{x_1},\dots,\dB_{x_k}$ where $k=\SB$). 
Let $\Lambda: \Gamma\rightarrow\mathbb{N}$ be the function defined by:
$$\Lambda(\gamma)=\sum_{v\in V}\lambda(\gamma,v).$$
Let $\phi:\Gamma\times V\rightarrow \mathbb{N}$ be the function defined by:
$$\phi(\gamma,v)=
\left\{
  \begin{array}{lrr}
    |\dis_{\p_v}-\dis_v-1| &\text{if} &\dis_{\p_v}<\SB_v   \\
    |\dis_v-1| &\text{if} & \dis_{\p_v}=\SB_v\\
  \end{array}
\right.
$$
Let $\Phi: \Gamma\rightarrow\mathbb{N}$ be the function defined by
$$\Phi(\gamma)=\sum_{v\in V}\phi(\gamma,v).$$ 
We define 
\[
\Gamma_{\tt CF}=\{\gamma\in \Gamma : \Psi(\gamma)=\Phi(\gamma)=\Lambda(\gamma)=0 \; \text{and} \; L(\gamma)>0\}.
\]

\begin{lemma}
\label{lem:CF}
$\Gamma_{\tt TEF} \triangleright \Gamma_{\tt CF}$ in $O(n\log n)$ rounds.
\end{lemma}

\begin{proof} 
Let us consider an initial configuration $\gamma_0\in \Gamma_{\tt TEF}$ such that the overlay structure induced by the parent variables $\p_v$, $v\in V$, forms a cycle $C$. The cycle $C$ necessarily contains all nodes, which implies that all nodes have non empty pointers parent. Moreover, since $\gamma_0\in \Gamma_{\tt TEF}$, we get that, for every nodes $v$, $\leader_v=0$. Thus $L(\gamma_0)=0$. 

Let $\xi:\Gamma \times V \rightarrow \mathbb{N}$ be the function defined by:
$$\xi(\gamma,v)=|\maxSB-\SB_v|$$
for every $\gamma\in\Gamma$, and every $v\in V$, where $\maxSB=\max \{ \SB_v \mid v\in \V\}$. Let then $\Xi:\Gamma\rightarrow\mathbb{N}$ be the function defined by,
$$\Xi(\gamma)=\mathlarger{\sum}_{v\in V} \xi(\gamma,v)$$ 
We define 
\[
\Gamma_\Xi=\{\gamma\in \Gamma : \Xi(\gamma)=0\}.
\]

\noindent \emph{Claim:} $\Xi(\gamma_0)= 0$. 
\medskip

Let us assume, for the purpose of contradiction, that $\Xi(\gamma_0)\neq 0$.  Let 
$$X=\{ v \in V \mid \SB_v=\maxSB  \}.$$ Since $\Xi(\gamma_0)\neq 0$, we have  $X\neq V$.
Moreover, since $C$ contains all nodes, we have that $\p_v\neq \emptyset$ for every $v\in V$, and there is at least one node $x \in X$ such that $\p_x \notin X$. This node $x$ satisfies $\SB_x \neq \SB_{\p_x}$. In this case, $\inNd(x)=false$ and  $\E_{\tt Nd}(x)=true$, which contradicts the fact that $\gamma_0\in \Gamma_{\tt TEF}$. Therefore, $\Xi(\gamma_0)= 0$.  

\medskip

Now, let $\pi:\Gamma\times V\rightarrow \mathbb{N}$ be the function defined by:
\[
\pi(\gamma,v)=|\Bs_v-\maxSB|+|\Phase_v-{\tt MaxPh}|+|\Bp_v-{\tt maxBp}|+|\Bc_v-{\tt maxC}|
\]
where 
\begin{eqnarray*}
{\tt MaxPh} & = & \max \{ \Phase_v , v\in \V\} \\
{\tt MaxPhNd} & =& \{ v\in \V \mid \Phase_v= {\tt MaxPh}\} \\
{\tt MaxBp} & = & \max \{\Bp_v , v\in {\tt MaxPhNd} \} \\
{\tt MaxBpNd} & = & \{ v\in {\tt MaxPhNd} \mid \Bp_v={\tt MaxBp} \} \\
{\tt MaxC} & = & \max \{ \Bc_v , v\in {\tt MaxPhNd}\}.
\end{eqnarray*}
Beside, let us define the quadruplet $${\tt MaxElec}=(\maxSB,{\tt MaxPh},{\tt MaxBp},{\tt MaxC}),$$
and the function  $\Pi:\Gamma\rightarrow\mathbb{N}$ such that 
$$\Pi(\gamma)=\mathlarger{\sum}_{v\in V}\pi(\gamma,v).$$
We define 
\[
\Gamma_\Pi=\{\gamma\in\Gamma : \Pi(\gamma)=0\}.
\] 

\noindent \emph{Claim:} $\Pi(\gamma_0)=0$.
\medskip 

 Recall that $\gamma_0 \in \Gamma_{\tt TEF}$, $\Xi(\gamma_0)=0$, and $L(\gamma_0)=0$. Assume for the purpose of contradiction that $\Pi(\gamma_0)\neq 0$.
Let $$X=\{v\in V \mid\ME_v={\tt MaxElec}  \}.$$
A direct consequence of $\Pi(\gamma_0)\neq 0$ is that $X\neq V$ and $|X|<n$.
Since the cycle $C$ contains all nodes, we have that, for every $v\in V$, $\p_v\neq\emptyset$, and there is at least one node $x\in X$ such that $\p_x \notin X$. We denote by $y=\p_x$ the parent of $x$, and we consider the variables $\Bs_x$, $\Phase_x$, $\Bp_x$, and $\Bc_x$.

We have $\Bs_x = \SB_x$ because otherwise $\inNd(v)=false$, in contradiction with $\gamma_0 \in \Gamma_{\tt TEF}$. Therefore $\Bs_x=\Bs_y$ because $\Xi(\gamma_0)=0$. 

If $| \Phase_x-\Phase_y|>1$  then $x$ or $y$ detects an error (see the predicate $\E_{\tt Phase}(v)$), and thus $\gamma_0 \notin \Gamma_{\tt TEF}$. If $| \Phase_x-\Phase_y| = 1$, then, since $y \notin X$, we have $\Phase_x=\Phase_y+1$. Thus $y$ cannot be a parent of $x$ (see the predicate ${\tt Coh_p}$). Thus $\gamma_0 \notin \Gamma_{\tt TEF}$. Therefore, $\Phase_x=\Phase_y$.

We now turn our attention to the variable $\Bp_x$. If $|X|=1$ then $\E_{\tt Bp}(x)=true$ because, on the one hand, $x$ has no neighbors $z$ such that $\Phase_z=\Phase_x$ and $\Bp_z=\Bp_x$, and, on the other hand, $x$ has no neighbors $z$ such that $\Phase_z=\Phase_x+1$ and $z=\Best_x$ (see predicate $\Best$). Thus, if $|X|=1$ then $\gamma_0 \notin \Gamma_{\tt TEF}$. Therefore, $|X|>1$. Since  $\Phase_x=\Phase_y$ we have  $\Bp_x = \Bp_y$ because otherwise $y$ cannot be a parent of $x$ (see the predicate ${\tt Coh_p}$). Therefore $\Bp_x=\Bp_y$. 

Finally, we consider $\Bc_x$, and show $\Bc_x=\Bc_y$. Assume that $\Bc_x\neq \Bc_y$. Then $|X|>1$ because otherwise Predicate  $\E_{\tt Control}(x)$ would be true, in contradiction with  $\gamma_0 \in \Gamma_{\tt TEF}$. Moreover, if $\Bc_y=1$ and $\Bc_x=0$, then $\E_{\tt Control}(x)=true$ and $\E_{\tt Control}(y)=true$, again contradicting  $\gamma_0 \in \Gamma_{\tt TEF}$. If  $\Bc_y=0$ and $\Bc_x=1$ then let $x'\in X$ be a descendent of $x$ whose  child $x''\notin X$. In this case, $\Bc_{x'}=1$ and $\Bc_{x''}=0$, hence Predicate  $\E_{\tt Control}(x')=true$, in contradiction with $\gamma\in \Gamma_{\tt TEF}$. Therefore $\Bc_x=\Bc_y$.

Since $\Bs_x=\Bs_y$, $\Phase_x=\Phase_y$, $\Bp_x=\Bp_y$, and $\Bc_x=\Bc_y$, we obtain that $y \in X$, in contradiction with $\p_x\notin X$. As a consequence, $\Pi(\gamma_0)=0$, which complete the proof of the claim. 

\medskip

An important consequence of  $\Pi(\gamma_0)=0$ is that, for every node $v$, we have  $\ME_v= {\tt MaxElec}$. 

\medskip
\noindent \emph{Claim:} $\Phi(\gamma_0) = 0$. 
\medskip

Again, the proof is by contradiction, assuming  $\Phi(\gamma_0) \neq 0$. Since $L(\gamma)=0$, there are no candidates for being the root of the tree. Thus Predicate $\TInc$$(v)=false$ for every node $v$, and therefore the command $\Inc(v)$ cannot be executed at any node $v$. As a consequence, all nodes are passive (i.e., $\forall v\in\V: \inNd(v)=true$). In addition, since $\ME_v= {\tt MaxElec}$  for every  $v\in V$, we also get that  the commands $\Passive(v)$ and $\Update(v)$ cannot be executed  at any node $v$. Since $\Phi(\gamma_0)\neq 0$, there exists at least one passive node $v$ that detects an error between its distance and the distance of its parent (see predicate $\E_{\tt d}(v)$). Hence, for that node $v$, the predicates $\inNd(v)$ and $\E_{\tt Nd}(v)$ are both true, which is a contradiction with $\gamma_0 \in \Gamma_{\tt TEF}$. Thus $\Phi(\gamma_0) = 0$. 

\medskip

We are now ready to show that, if the initial configuration $\gamma_0$ contains a cycle, then Algorithm $\Algo$ detects an error in $O(n\log n)$ rounds. 

Since $\Phi(\gamma_0)=0$, we necessarily have that $n$ is a multiple of $\maxSB$, and that there are $n/ \maxSB$ hyper-nodes. 
Since all nodes are passive in $\gamma_0$, the only commands that can be executed by some node(s) are related to the distance verification between hyper-nodes, that is Commands
$\BinaryAd(v), \Broadcast(v), \Verification(v)$ and $ \CleanM(v)$. 
More specifically, the only nodes that can possibly be activated in $\gamma_0$ are the nodes $v$ such that $\dis_v=\SB_v$. 

For every hyper-node $X=(x_1,x_2,\dots,x_k)$, where $k=\maxSB$,  since the scheduler is weakly fair, predicate $\TAdd(x_k)= true$, and $x_k$ executes the command $\BinaryAd(x_k)$ at round $1$. This yields the execution of the binary addition. 
The binary addition occurs from $x_k$ to $x_1$, and every node in each hyper-node $X$ eventually takes value ''$+$'' or ''$ok$'' once $\maxSB$ rounds has been performed. Now, if $\dB_{x_1}=1$ and $\Add_{x_1}=+$, then an error is detected since the binary addition overflows  beyond the limit of $\maxSB$ bits (see $\E_{\tt MAdd}(v)$).

Node $x_1$ starts the verification process that propagates from $x_1$ to $x_k$. Fix any hyper-node $X=(x_1,x_2,\dots,x_k)$, and let us denote by $Y$ the hyper-node child of $X$ in the current configuration at round $\maxSB$. Node $x_1$ computes the values of $\dB_{y_1}$  (see the  predicate $\TBroad$, and the command  $\Broadcast$).  This value is broadcast from $x_1$ to $x_{k}$ (see the predicate $\TBroad$, and the command $\Broadcast$). Node $y_1$ checks whether $\MC_{x_{k}}[1]=\dB_{y_1}$. If it is the case, then the verification process for all other nodes in $Y$ carries on (see the predicate  $\TVerif$ and the command $\Verification$). Otherwise $y_1$ detects an error (see $\E_{\tt  Hyper}(v)$).  
%
%
Thanks to Predicates  $\TAdd$, $\TBroad$, and $\TVerif$, and to commands $\BinaryAd$, $\Broadcast$, and $\Verification$ every node in $Y$ are eventually checked, after an additional $\maxSB$ rounds. 
The total number of rounds for checking hyper-nodes is the following: there are $n/\lfloor \log n\rfloor$ hyper-nodes, each hyper-node performs verification in $O(\lfloor \log n\rfloor)$ rounds, so the overall process takes $O(n \log n)$ rounds.

To sum-up, if  the configuration $\gamma_0$ contains a cycle, then at least one hyper-node detects an error in a $O(n\log n)$ rounds.

Finally, if $L(\gamma_0)=0$ in an arbitrary configuration, then the algorithm $\Algo$ detects an error in $O(n \log n)$ rounds, and, by Lemma~\ref{lem:TEF}, the system reaches a configuration in $\Gamma_{\tt CF}$.
\end{proof}


\begin{lemma}
\label{lem:CF-Close}
$\Gamma_{\tt CF}$ is closed.
\end{lemma}

\begin{proof}
Let us consider a configuration $\gamma \in \Gamma_{\tt CF}$.
We already noticed in the proof of Lemma~\ref{lem:TEF-close} that Algorithm \Algo\/ preserves coherent distances  (i.e., $\Phi(\gamma)=0$), and does not introduce trivial errors  (i.e., $\Psi(\gamma)=0$).
Moreover, in the proof of Lemma~\ref{lem:CF}, we have explained that the hyper-node distance verification correctly reports errors, if any. The variable $\dB$ is only modified by the command $\Passive$. In the sequel, we use the wording ``$v$ joins $T_r$'' when a node $v$ executes the command $\Passive$, and the pointer $\p_v$ of $v$ points to a node in the subtree $T_r$ rooted at $r$. 

We denote by $X$ the set of candidate nodes. Let us first consider a node $v$ (in $X$ or not), and a root $r\in X$ such that $\dis_v\leq \lfloor \log n\rfloor$, and $v$ joins $T_r$ when $r$ increases its phase from $i-1$ to $i$, for some $i$. Thanks to Predicate $\TStartdB$, and to command $\StartdB$, node $r$ publishes first the bit for the node at distance~1 from $r$. In other words, $\MC_r=(1,0)$. When any node $v$ at distance~1 joins $T_r$, $v$ sets $\dB_v=\MC_r[1]$, and then $v$ informs $r$ about the updating of $\dB$ by setting  $\MC_v=\MC_r$. At this point,  $r$ can publish the bit for nodes at distance~2 (i.e., $\MC_r=(2,0)$), and so on until the distance reaches $\lfloor \log n\rfloor$. 
Now, a node $v$ joins $T_r$ only if its candidate parent publishes the bit that corresponds to the binary representation of the distance between $v$ and $r$. In other words, for any node $u$, if $\dis_u=k$ with $k<\lfloor \log n\rfloor$ then $\Bs_u$ must be equal to $k+1$. This enables to maintain $\Lambda(\gamma)=0$. 
When $k$ reaches $\lfloor \log n\rfloor$, a hyper-node is created. Then, a binary addition process is carried out, and computes the bit for $v$ when $v$ joins $T_r$. This process maintains $\Lambda(\gamma)=0$, and, as a direct consequence $L(\gamma)$ remains greater than~$0$. To conclude, algorithm \Algo\/ keeps $\Psi(\gamma)=0$, $\Phi(\gamma)=0$, $\Lambda(\gamma)=0$, and $L(\gamma)>0$.
\end{proof}

We now introduce Predicate $\Gamma_{\tt IEF}$ (for \emph{Impostor Error Free}), which ensures that the currently elected leader is not an impostor.
Let $\rho:\Gamma\times V \rightarrow \mathbb{N}$ be the function defined by:
$$\rho(\gamma,v)=|{\tt maxFB}-\SB_v|$$
where ${\tt maxFB}=\max \{\Bit(1,\id_v) \mid v\in V\}$.  Let $\Rho:\Gamma\rightarrow\mathbb{N}$ be the function defined by
$$\Rho(\gamma)=\sum_{v\in V}\rho(\gamma,v).$$ 
We show that $\Algo$ reaches a legitimate configuration with respect to leader election if and only if $\Rho(\gamma)=0$. 

Let $\epsilon:\Gamma\times V \rightarrow \mathbb{N}$ be the function defined by:
$$\epsilon(\gamma,v)=\left\{
  \begin{array}{lll}
  0 &\text{if} & \dis_v=0 \wedge \Bit({\tt minPh},\id_v)=\Bit({\tt minPh},l^*)\\
  1&\text{if} & \dis_v=0 \wedge \Bit({\tt minPh},\id_v)<\Bit({\tt minPh},l^*)\\
    0&\multicolumn{2}{l}{\text{otherwise} }   \\
  \end{array}
\right.
$$
where ${\tt minPh}=\min \{\Phase_v \mid v \in V\}$ and $l^*$ is the identity of the node with the maximum identity.
 Let ${\tt E}:\Gamma\rightarrow\mathbb{N}$ be the function defined by,
$${\tt E}(\gamma)=\sum_{v\in V} \epsilon(\gamma,v).$$ 
We define 
\[
\Gamma_{\tt IEF}=\{\gamma\in\Gamma : \Psi(\gamma)=\Rho(\gamma)={\tt E}(\gamma)=0\}.
\]

\begin{lemma}
\label{lem:IEF}
$ \Gamma_{\tt CF}\triangleright \Gamma_{\tt IEF}$ in $O(n\log n)$ rounds.
\end{lemma}

\begin{proof}
Let us consider first an initial configuration $\gamma\in \Gamma_{\tt CF}$. We have observed in lemma~\ref{lem:CF} that being in $\Gamma_{\tt CF}$ implies $L(\gamma)>0$. Let us denote by $\ell^*$ the node with maximum identity, and $X$ the set of the candidate nodes. Let us suppose that $\ell^*\not \in X$. Let us denote by $\ell$ the node with the maximum identity in $X$. In the worst case, all the sub-spanning tree merge in a unique spanning tree rooted at $\ell$. Thus, let us suppose that $\gamma$ is a configuration where the network is spanned by an unique tree rooted in $\ell$.
In this case $\dis_\ell=0$ and $\dis_u>0$ for every node $u\neq \ell$. 

Let us assume, for the purpose of contradiction, that $\Rho(\gamma)\neq 0$. If the tree is rooted at $\ell$, then every node must have the same $\SB$ as $\ell$. Since $\ell$ is a root, $\SB_\ell=\Bit(1,\id_\ell)$. Hence, $\SB_\ell\neq  {\tt maxFB}$ (because $\Rho(\gamma)\neq0$). Now, $\SB_\ell$ cannot be larger than ${\tt maxFB}$, so there exists $v$ such that $\SB_v={\tt maxFB}$, and $\inNd(v)$ is true. This contradicts $\gamma\in \Gamma_{\tt TEF}$, so we can conclude that $\Rho(\gamma)=0$.

 Since $\ell$ and $\ell^*$ have the same number of bits, there must exist one phase where the bit-position of $\ell^*$ is larger than the bit-position of $\ell$. More formally, there exists $i$, $1<i\leq \lfloor \log n \rfloor +1$, such that, for every $j<i$, we have $\Bit(j,\id_\ell)=\Bit(j,\id_{\ell^*})$, and, for every $k\geq i$, we have $\Bit(k,\id_\ell)<\Bit(k,\id_{\ell^*})$. Note that $i>1$, because, in $\Gamma_{\tt TEF}$, predicate $\inNd$ must be true. The worst case with respect to time complexity is for $i=\lfloor \log n \rfloor +1$, and the arbitrary initial configuration starts at phase $2$. In this case, only $\R_{\tt Root}$ can be activated for $\ell$, and only the rule $\R_{\tt Node}$ for the other nodes (in parallel to the hyper-node distance verification). Node $\ell$ executes the command $\Inc(\ell)$ for $\lfloor \log n \rfloor +1 -2$ times. After each execution of command $\Inc$, every node executing this command updates the variable $\ME$ in a top-down manner (see Predicate $\TUp$ and command  $\Update$). This updating process takes at most $n$ rounds. When all nodes have the same election values, a bottom-up control process is initiated (see Predicate $\TUp$ and command $\Update$). This process takes at most $n$ rounds. After that, $\ell$ increases its phase, and the same process is repeated. At the last phase, $\E_{\tt Elec}(\ell^*)$ is true, and an error is detected. Therefore, if the system contains a impostor leader, then an error is detected in $O(n \log n)$ rounds.
\end{proof}

\begin{lemma}
\label{lem:IEF-Close}
$\Gamma_{\tt IEF}$ is closed.
\end{lemma}

\begin{proof}
Let $\gamma \in \Gamma_{\tt IEF}$, with $L(\gamma)>0$. Let $X$ be the set of candidate leaders (i.e., for every $x\in X$,  $\Root(x)=true$). Let $x\in X$, and let $T_{(x,i)}$ be the subtree rooted at $x$ during phase $i$. For a node $v\in T_{(x,i)}$, we have that (1) either  $\Bs_v<\Bs_{x}$ or $\Bs_v=\Bs_{x}$,  and (2) $\Bit(j,\id_v)<\Bit(j,\id_x)$ for every $j\leq i$. Moreover, for any two candidates leaders $x_1$ and $x_2$, we have that, at phase $i$, $\Bs_{x_1}=\Bs_{x_2}$, and $\Bit(j,\id_{x_1})=\Bit(j,\id_{x_2})$ for evry $j\leq i$. Let $i$ be such that $T_{(x_1,i)}$ and $T_{(x_2,i)}$ have adjacent nodes. Let $x'_1\in T_{(x_1,i)}$ and $x'_2\in T_{(x_2,i)}$ be two nodes such that $x'_1$ is adjacent to $x'_2$. If, at phase $i+1$, $\Bit(i+1,\id_{x_1})>\Bit(i+1,\id_{x_2})$, then the nodes of $T_{(x_1,i)}$ and $T_{(x_2,i)}$ activated by Predicate $\TUp$ executes $\Update$. When node $x'_1$ reaches $\ME_{x_1}$ at phase $i+1$, node $x'_2$ becomes passive (cf. command $\Passive$) and selects node $x'_1$ as its parent. In a bottom-up fashion, every node of $T_{(x_2,i)}$ joins the subtree $T_{(x_1,i+1)}$ (cf. command $\Passive$), and, eventually, $x_2$ becomes passive and joins $T_{(x_1,i+1)}$. By this process, for every node $v$ in  $T_{(x_1,i+1)}$, we have (1) either $\Bs_v<\Bs_{x_1}$ or $\Bs_v=\Bs_{x_1}$,  and (2) $\Bit(j,\id_v)<\Bit(j,\id_{x_1})$ for every $j\leq i+1$. This process is repeated until  phase $\lfloor \log n\rfloor$, where there remains a single leader in the network.
\end{proof}

Note that if a node executes $\Reset$, and since the scheduler is weakly fair, at most $n$ rounds later  all nodes have executed $\Reset$. 

\begin{lemma}
\label{lem:L}
$\Gamma_{\tt IEF}\triangleright \Gamma_{\tt LE}$ in $O(n\log^2 n)$ rounds.
\end{lemma}

\begin{proof}
We know by Lemma~\ref{lem:CF-Close} that $\Gamma_{\tt CF}$ is closed, and we know by Lemma~\ref{lem:IEF-Close} that  $\Gamma_{\tt IEF}$ is closed. Moreover the proof of Lemma~\ref{lem:IEF-Close} provided details about the election process at each phase. Let $\gamma \in \Gamma_{\tt IEF}$ at round $t$. Moreover, let us suppose that $\gamma\in\Gamma_{\tt CF}$. That is, $L(\gamma)>0$. More precisely, 
let $i$, $1\leq i\leq \lfloor \log n \rfloor +1$, denote the smallest phase counter in the network, among all nodes. At phase $i$, there are at most $n/2^{i-1}$ candidate leaders (i.e., at most this many roots). Thus, $L(\gamma)=n/2^{i-1}$ at phase $i$. We have studied in the proof of Lemma~\ref{lem:IEF} how Algorithm $\Algo$ performs the election process. After $O(n \log n)$ rounds, Phase $i+1$ is completed, and the system reaches some configuration $\gamma'$. At this point, there remain at most $n/2^{i}$ candidate leaders. Since $L(\gamma')\leq n/2^{i}$, we get that 
\[
L(\gamma')<L(\gamma).
\]
The number of phases is upper bounded by $\lfloor \log n\rfloor+1$. At phase $\lfloor \log n\rfloor+1$, we reach a configuration $\gamma''$ satisfying $L(\gamma'')=1$.
A direct consequence of Lemma~\ref{lem:IEF} and Lemma~\ref{lem:IEF-Close} is that only the node $\ell^*$ with  maximum identity has $\leader_{\ell^*}=1$. Every other node $v$ has $\leader_{v}=0$. Moreover, for every node $v\neq \ell^*$, we have $\p_v\neq \emptyset$, and the structure induced by the pointers  $\p_v$, for all $v\neq\ell^*$ forms a spanning tree rooted in $\ell^*$.
Regarding time complexity, our algorithm takes $O(n \log n)$ rounds to detect an impostor leader, $O(n)$ rounds to reset the system after the detection of an error, and $O(n \log^2 n)$ rounds to elect the leader. Therefore, in total,  Algorithm $\Algo$ performs $O(n \log^2 n)$ rounds to converge to the leader specification. 
\end{proof}

\begin{lemma}
\label{lem:L-Close}
$\Gamma_{\tt LE}$ is closed.
\end{lemma}

\begin{proof}
The rule $\R_{\tt Passive}(v)$ is the only rule performed by $v$ that modifies the distance and the leader variables of node $v$. Let $\ell^*$ be the node with the maximum identity. As a direct consequence of Lemma~\ref{lem:L}, in the initial configuration, $\ell^*$ is the only node that has  $\dis_{\ell^*}=0$ and $\leader_{\ell^*}=1$. In other words $\ell^*$ is the only elected node. Therefore, $\ell^*$changes the phase of the system by increasing the current phase, or by restarting from phase~1 (see predicate $\TInc$ and command $\Inc$). Thus, every node $v$ satisfies $\Bs_v\leq\Bs_{\ell^*}$. Moreover, for every phase $i$, $1\leq i\leq \lfloor \log n\rfloor+1$, every node $v$ satisfies $\Bit(i,\id_v)<\Bp_{\ell^*}$. Hence, every node can only executes the command $\Update$. Finally, nodes never change their distance, parent,  and leader variables.
\end{proof}

\subsection{Memory requirements}

\begin{lemma}
\label{lem:Memorie}
Algorithm \Algo\/ use $O(\log \log n)$ bits of memory per node.
\end{lemma}

\begin{proof}
Algorithm \Algo\/ has two types of variables: the variables that use a constant number of bits, and those that use $O(\log \log n)$ bits. Variables of the first type are: 
\[
\p_v\in\{\emptyset,0,1\}, \;\; \dB_v\in \{0,1\}, \;\; \Add_v\in\{+,ok,\emptyset\}, \; \text{and} \; \leader_v\in \{0,1\}. 
\]
Variables of the second type are: 
\[
\SB_v\in\{1,...,\lfloor \log n \rfloor\}, \;\; 
\MC_v\in \{1,...,\lfloor \log n \rfloor\}\times \{0,1\}, \;\;
\MV_v\in \{1,...,\lfloor \log n \rfloor\}\times\{0,1\},
\]
and
\[
\ME_v\in \{1,...,\lfloor \log n \rfloor\}\times\{1,...,\lfloor \log n \rfloor\}\times\{1,...,\lfloor \log n \rfloor\}\times\{0,1\}.
\]
Hence, \Algo\/ uses $O(\log \log n)$ bits of memory per node.
\end{proof}

\section{Conclusion}

In this paper, we have shown that, in the state model, with a weakly fair distributed scheduler, one can elect a leader in a ring with a (non-silent) self-stabilizing algorithm using only $O(\log\log n)$ bits of memory per node. It is known that one cannot do the same using only $O(1)$ bits of memory per node (see~\cite{BGJ07j}). An intriguing question is whether one can perform leader election in the same framework as in this paper, using just $o(\log\log n)$ bits per node, and, if yes, to what extend can the memory requirement for (non silent) leader election being reduced. A natural candidate function for the minimum memory requirement for leader election is $O(\log^*n)$ bits per node, by applying the techniques in this paper recursively. This however seems to be non trivial, as self-stabilization has to be maintained at every level of the recursion. 

\bibliographystyle{plain}	
\bibliography{BibDISC}

\begin{thebibliography}{10}

\bibitem{ANT12c}
J.~Adamek, M.~Nesterenko, and S.~Tixeuil.
\newblock Using abstract simulation for performance evaluation of stabilizing
  algorithms: The case of propagation of information with feedback.
\newblock In {\em SSS 2012}, LNCS. Springer, 2012.

\bibitem{AB98j}
Y.~Afek and A.~Bremler-Barr.
\newblock Self-stabilizing unidirectional network algorithms by power supply.
\newblock {\em Chicago J. Theor. Comput. Sci.}, 1998.

\bibitem{AG94j}
A.~Arora and M.~G. Gouda.
\newblock Distributed reset.
\newblock {\em IEEE Trans. Computers}, 43(9):1026--1038, 1994.

\bibitem{AK09c}
M.~Arumugam and S.~S. Kulkarni.
\newblock Prose: A programming tool for rapid prototyping of sensor networks.
\newblock In {\em S-CUBE}, pages 158--173, 2009.

\bibitem{AKMPV07j}
B.~Awerbuch, S.~Kutten, Y.~Mansour, B.~Patt-Shamir, and G.~Varghese.
\newblock A time-optimal self-stabilizing synchronizer using a phase clock.
\newblock {\em IEEE Trans. Dependable Sec. Comput.}, 4(3):180--190, 2007.

\bibitem{AO94c}
B.~Awerbuch and R.~Ostrovsky.
\newblock Memory-efficient and self-stabilizing network reset.
\newblock In {\em PODC}, pages 254--263. ACM, 1994.

\bibitem{BDDT07j}
J.~Beauquier, S.~Dela\"{e}t, S.~Dolev, and S.~Tixeuil.
\newblock Transient fault detectors.
\newblock {\em Distributed Computing}, 20(1):39--51, June 2007.

\bibitem{BGJ07j}
J.~Beauquier, M.~Gradinariu, and C.~Johnen.
\newblock Randomized self-stabilizing and space optimal leader election under
  arbitrary scheduler on rings.
\newblock {\em Distributed Computing}, 20(1):75--93, January 2007.

\bibitem{BlinT_PODC13}
L.~Blin and S.~Tixeuil.
\newblock Brief announcement: deterministic self-stabilizing leader election
  with o(log log n)-bits.
\newblock In {\em Proceedings of the 32st ACM Symposium on Principles of
  Distributed Computing, (PODC13)}, pages 125--127, 2013.

\bibitem{BlinT_DISC13}
L.~Blin and S.~Tixeuil.
\newblock Compact deterministic self-stabilizing leader election: The
  exponential advantage of being talkative.
\newblock In {\em Proceedings of the 27th International Conference on
  Distributed Computing (DISC 2013)}, Lecture Notes in Computer Science (LNCS),
  pages 76--90. {S}pringer {B}erlin / {H}eidelberg, 2013.

\bibitem{CG12j}
Y.~Choi and M.~G. Gouda.
\newblock A state-based model of sensor protocols.
\newblock {\em Theor. Comput. Sci.}, 458:61--75, 2012.

\bibitem{DMCDHSHLAG08c}
A.~R. Dalton, W.~P. McCartney, K.~Ghosh Dastidar, J.~O. Hallstrom, N.~Sridhar,
  T.~Herman, W.~Leal, A.~Arora, and M.~G. Gouda.
\newblock Desal alpha: An implementation of the dynamic embedded
  sensor-actuator language.
\newblock In {\em ICCCN}, pages 541--547. IEEE, 2008.

\bibitem{DLV11j}
A.~Kumar Datta, L.~L. Larmore, and P.~Vemula.
\newblock Self-stabilizing leader election in optimal space under an arbitrary
  scheduler.
\newblock {\em TCS}, 412(40):5541--5561, 2011.

\bibitem{DMT09c}
S.~Devismes, T.~Masuzawa, and S.~Tixeuil.
\newblock Communication efficiency in self-stabilizing silent protocols.
\newblock In {\em ICDCS 2009}, pages 474--481. {IEEE} Press, 2009.

\bibitem{D74j}
E.~W. Dijkstra.
\newblock Self-stabilizing systems in spite of distributed control.
\newblock {\em Commun. ACM}, 17(11):643--644, 1974.

\bibitem{D00b}
S.~Dolev.
\newblock {\em Self-stabilization}.
\newblock MIT Press, March 2000.

\bibitem{DGS99j}
S.~Dolev, M.~G. Gouda, and M.~Schneider.
\newblock Memory requirements for silent stabilization.
\newblock {\em Acta Inf.}, 36(6):447--462, 1999.

\bibitem{DH97j}
S.~Dolev and T.~Herman.
\newblock Superstabilizing protocols for dynamic distributed systems.
\newblock {\em Chicago J. Theor. Comput. Sci.}, 1997, 1997.

\bibitem{DIM97j}
S.~Dolev, A.~Israeli, and S.~Moran.
\newblock Resource bounds for self-stabilizing message-driven protocols.
\newblock {\em SIAM J. Comput.}, 26(1):273--290, 1997.

\bibitem{DT11r}
S.~Dubois and S.~Tixeuil.
\newblock A taxonomy of daemons in self-stabilization.
\newblock Technical Report 1110.0334, ArXiv eprint, October 2011.

\bibitem{FJ01c}
F.~E. Fich and C.~Johnen.
\newblock A space optimal, deterministic, self-stabilizing, leader election
  algorithm for unidirectional rings.
\newblock In {\em DISC}, pages 224--239. Springer, 2001.

\bibitem{GCH06c}
M.~G. Gouda, J.~Arturo Cobb, and C.~Huang.
\newblock Fault masking in tri-redundant systems.
\newblock In {\em SSS}, LNCS, pages 304--313. Springer, 2006.

\bibitem{HP00j}
T.~Herman and S.~V. Pemmaraju.
\newblock Error-detecting codes and fault-containing self-stabilization.
\newblock {\em Inf. Process. Lett.}, 73(1-2):41--46, 2000.

\bibitem{H98j}
J.~Hoepman.
\newblock Self-stabilizing ring-orientation using constant space.
\newblock {\em Inf. Comput.}, 144(1):18--39, 1998.

\bibitem{IJ93j}
A.~Israeli and M.~Jalfon.
\newblock Uniform self-stabilizing ring orientation.
\newblock {\em Inf. Comput.}, 104(2):175--196, 1993.

\bibitem{IL94c}
G.~Itkis and L.~A. Levin.
\newblock Fast and lean self-stabilizing asynchronous protocols.
\newblock In {\em FOCS}, pages 226--239. IEEE Computer Society, 1994.

\bibitem{ILS95c}
G.~Itkis, C.~Lin, and J.~Simon.
\newblock Deterministic, constant space, self-stabilizing leader election on
  uniform rings.
\newblock In {\em WDAG}, LNCS, pages 288--302. Springer, 1995.

\bibitem{KormanKM11}
Amos Korman, Shay Kutten, and Toshimitsu Masuzawa.
\newblock Fast and compact self stabilizing verification, computation, and
  fault detection of an mst.
\newblock In {\em Proceedings of the 30th Annual ACM Symposium on Principles of
  Distributed Computing, PODC 2011}, pages 311--320, 2011.

\bibitem{MT09j}
T.~Masuzawa and S.~Tixeuil.
\newblock On bootstrapping topology knowledge in anonymous networks.
\newblock {\em ACM Transactions on Adaptive and Autonomous Systems}, 4(1),
  2009.

\bibitem{MOOY92c}
A.~J. Mayer, Y.~Ofek, R.l Ostrovsky, and M.~Yung.
\newblock Self-stabilizing symmetry breaking in constant-space (extended
  abstract).
\newblock In {\em STOC}, pages 667--678, 1992.

\bibitem{MG05b}
T.~M. McGuire and M.~G. Gouda.
\newblock {\em The Austin Protocol Compiler}, volume~13 of {\em Advances in
  Information Security}.
\newblock Springer, 2005.

\bibitem{T09bc}
S.~Tixeuil.
\newblock {\em Algorithms and Theory of Computation Handbook}, pages
  26.1--26.45.
\newblock CRC Press, Taylor \& Francis Group, 2009.

\end{thebibliography}


\end{document}